\documentclass[12pt]{article}
\usepackage{amsmath}
\usepackage{amssymb}
\usepackage{amsthm}
\usepackage{graphicx}
\usepackage{lineno}
\usepackage{float}
\usepackage[all]{xy}
\usepackage{mathrsfs}
\usepackage{tikz-cd}

\usepackage{xcolor}
\usepackage{tikz}
\usetikzlibrary{tqft}

\newtheorem{theorem}{Theorem}

\newtheorem{definition}{Definition}

\newcommand{\be}{\begin{equation}}
\newcommand{\ee}{\end{equation}}

\newcommand{\bse}{\begin{subequations}}
\newcommand{\ese}{\end{subequations}}
\newcommand{\bea}{\begin{eqnarray}}
\newcommand{\eea}{\end{eqnarray}}
\newcommand{\ba}{\begin{array}}
\newcommand{\ea}{\end{array}}
\newcommand{\bc}{\begin{center}}
\newcommand{\ec}{\end{center}}

\usepackage{braket}

\newcommand{\Gr}{{\rm Gr}}

\newcommand{\Tr}{{\rm Tr}}

\newcommand{\lra}{\longrightarrow}

\newcommand{\sfM}{\mathsf{M}}

\newcommand{\del}{\partial}

\newcommand{\medn}{\medbreak \noindent}

\begin{document}

%\title{\bf Universal quantum computation in topological quantum neural networks and amplituhedron representation}

\title{\bf Amplituhedra for generic quantum processes via the TQNN representation of UQC}

\author{{Chris Fields$^1$, James F. Glazebrook$^2$, Antonino Marcian\`{o}$^3$ and Emanuele Zappala$^4$}\\ \\
{\it$^1$ Allen Discovery Center}\\
{\it Tufts University, Medford, MA 02155 USA}\\
{fieldsres@gmail.com} \\
{ORCID: 0000-0002-4812-0744} \\ \\
{\it$^2$ Department of Mathematics and Computer Science,} \\
{\it Eastern Illinois University, Charleston, IL 61920 USA} \\
{\it and} \\
{\it Adjunct Faculty, Department of Mathematics,}\\
{\it University of Illinois at Urbana-Champaign, Urbana, IL 61801 USA}\\
{jfglazebrook@gmail.com}\\
{ORCID: 0000-0001-8335-221X}\\ \\
{\it$^3$ Center for Field Theory and Particle Physics \& Department of Physics} \\
{\it Fudan University, Shanghai, CHINA} \\
{\it and} \\
{\it Laboratori Nazionali di Frascati INFN, Frascati (Rome), ITALY} \\
{marciano@fudan.edu.cn} \\
{ORCID: 0000-0003-4719-110X} \\ \\
{\it$^4$ Department of Mathematics and Statistics,} \\
{\it Idaho State University, Pocatello, ID 83209, USA} \\
{emanuelezappala@isu.edu} \\
{ORCID: 0000-0002-9684-9441} \\ \\
}
\maketitle

\begin{abstract}
\noindent
We study the relationship between computation and scattering both operationally (hence phenomenologically) and formally.  We develop a representation of universal quantum computation (UQC) within the formalism of topological quantum neural networks (TQNNs), using the Reshetikhin-Turaev and Turaev-Viro models to show how TQNNs implement quantum error-correcting codes.  We then exhibit a formal correspondence between TQNNs and amplituhedra to support the existence of amplituhedra for representing generic quantum processes.  This construction shows how amplituhedra are geometric representations of underlying topological structures.  We conclude by pointing to applications areas enabled by these results.

\end{abstract}


\tableofcontents

\section{Introduction}

The relationship between quantum computation and physical processes such as scattering has been of interest since Feynman suggested, in 1982, that a sufficiently large quantum computer could exactly simulate any physical process \cite{feynman:82}.  This suggestion was proved correct by Lloyd in 1996 \cite{lloyd:96}.  The practical demands of particle physics have motivated the development of implementable quantum algorithms to simulate scattering processes, e.g. to calculate amplitudes and/or phase shifts \cite{elvang:15,henn:14,jordan:14,turro:24}.  On the other hand, the further development of theoretical models \cite{nielsen:00}, and physical implementations \cite{shor:94} of quantum computers has demonstrated that many physical processes can implement general, or ``universal'' quantum computation (UQC).  In particular, single-photon linear optics using the Knill-Laflamme-Milburn (KLM) protocol \cite{knill:01} can implement UQC, as can planar-graph, i.e. quantum walk representations of scattering, including the Bose-Hubbard \cite{childs:13} and Fermi-Hubbard \cite{bao:15} models.  Since it is straightforward to interpret any measurable physical process as computation \cite{fields:89, horsman:14}, the universality of scattering as an implementation of computation suggests that {\em any} physical process can be treated as scattering.  This accords with intuition: scattering is the most fundamental observable manifestation of physical interaction.

We can depict these intuitions with a diagram:

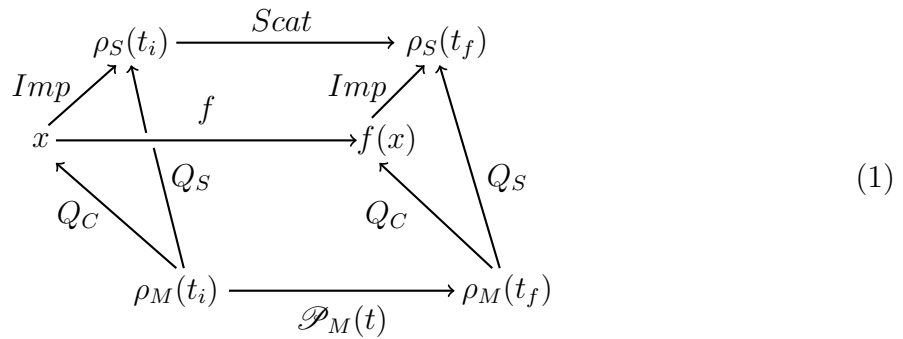
\begin{equation} \label{comp-scat}
\begin{gathered}
\begin{tikzpicture}
\node at (0,.4) {$f$};
\draw [thick, ->] (-2,0) -- (2,0);
\node at (-2.2,0) {$x$};
\node at (2.4,0) {$f(x)$};
\node at (-0.4,-2) {$\rho_M(t_i)$};
\node at (4,-2) {$\rho_M(t_f)$};
\draw [thick, ->] (0.3,-2) -- (3.3,-2);
\node at (1.8,-2.4) {$\mathscr{P}_M(t)$};
\draw [thick, ->] (-0.4,-1.7) -- (-2,-0.3);
\node at (-1.7,-1) {$Q_C$};
\draw [thick, ->] (3.8,-1.7) -- (2.3,-0.3);
\node at (2.4,-1) {$Q_C$};
\draw [thick] (-0.3,-1.7) -- (-0.7,-0.1);
\draw [thick, ->] (-0.8,0.1) -- (-1,1);
\node at (-0.2,-0.5) {$Q_S$};
\node at (-1,1.3) {$\rho_S(t_i)$};
\draw [thick, ->] (3.9,-1.7) -- (3.1,1);
\node at (4,-0.5) {$Q_S$};
\node at (3.2,1.3) {$\rho_S(t_f)$};
\draw [thick, ->] (-0.4,1.3) -- (2.5,1.3);
\node at (1,1.6) {$Scat$};
\draw [thick, ->] (-2.1,0.2) -- (-1.2,1);
\node at (-2.2,0.7) {$Imp$};
\draw [thick, ->] (2.2,0.3) -- (2.9,1);
\node at (2,0.7) {$Imp$};
\end{tikzpicture}
\end{gathered}
\end{equation}

Here $\mathcal{P}_M(t): \mathcal{H}_M \rightarrow \mathcal{H}_M$ is the time-propagator acting on the Hilbert space $\mathcal{H}_M$ of some effectively-isolated quantum system or ``machine'' $M$.  The operator $Q_C : \mathcal{H}_M \rightarrow \{0,1\}^N$ sends states $\rho_M(t)$ to $N$-bit strings, while the operator $Q_S : \mathcal{H}_M \rightarrow \mathcal{H}_S$ sends states $\rho_M(t)$ to state representations in some scattering model $S$, e.g. the KLM model.  Requiring that $fQ_C = Q_C \mathcal{P}_M(t)|_{t_i \rightarrow t_f}$ interprets $\mathcal{P}_M(t)|_{t_i \rightarrow t_f}$ as a computation of the function $f: \{0,1\}^N \rightarrow \{0,1\}^N$, which clearly must be Turing-computable to construct the mapping \cite{horsman:14}.  Requiring that the scattering process $Scat$ be such that the map $Q_S$ is well-defined renders it a description of $\mathcal{P}_M(t)|_{t_i \rightarrow t_f}$ as (typically multi-particle) scattering.  The claim that $Scat$ implements UQC is then just the claim that the implementation map $Imp$ is well-defined, i.e. that the entire diagram commutes. We will make the commutativity of Diagram~\eqref{comp-scat} more precise below in a categorical setting compatible with the LOCC protocols required for operational measurements.

Diagram \eqref{comp-scat} raises an obvious question.  All current evidence suggests that the Standard Model (SM) provides a correct description of most physical interactions in regimes where gravity can be ignored.  Choosing a map $Q_S = Q_{\mathrm{SM}}$ should, therefore, allow us to represent $\mathcal{P}_M(t)$ for almost any system $M$ as SM scattering.  The claim that Diagram \eqref{comp-scat} commutes is then the claim that for any computable $f$, a mapping $Imp$ exists that implements the computation of $f$ as SM scattering, i.e. that SM scattering implements UQC.  From this perspective, the generation of magic-state resources that can enable quantum-computational advantage by accelerators \cite{white:24} is to be expected.

Representing $\mathcal{P}_M(t)$ in planar $\mathcal{N}=4$ supersymmetric Yang-Mills (SYM) theory, and hence choosing $Q_S = Q_{\mathrm{SYM}}$ in Diagram \eqref{comp-scat}, raises a deeper question.  As shown by Arkani-Hamed and Trnka~\cite{arkani:14}, scattering processes in SYM can be associated with subspaces -- {\em amplituhedra} -- of particular positive Grassmanians, from which their amplitudes can be computed without resort to off-shell excitations or perturbation theory.  The claim that Diagram \eqref{comp-scat} commutes is, in this case, not only the claim that SYM scattering implements UQC, but also that any time-propagator $\mathcal{P}_M(t)$, associated with any quantum computation, for which the former is representable in SYM, has an associated amplituhedron from which its amplitudes can be computed; see \S \ref{implement} below for an explicit statement. Here, we address this second question from a general perspective.

As this work spans two traditionally-distinct areas of research that have each developed its own concepts, methods, and notation, we provide an ample amount of background material with necessary references.  In outline, we proceed as follows.  We begin in \S \ref{locc} by reviewing the characterization of physical processes as computations developed in \cite{horsman:14}, placing this in the operational setting of {\em Local Operations, Classical Communication} (LOCC) protocols \cite{chit:14}.  We review previous results \cite{fgm:22a} showing how the map $Q_C$ can always be identified with a quantum reference frame (QRF) \cite{bartlett:07} represented as an operator on a qubit (quantum bit) space,\footnote{The state of a qubit consists of all superpositions $\alpha \vert 0 \rangle + \beta \vert 1 \rangle$ (where $\vert \alpha \vert^2 + \vert \beta \vert^2 =1$) of the ``computational'' basis states $\vert 0 \rangle$ and $\vert 1 \rangle$.} and how sequential measurements with such QRFs can be associated with topological quantum field theories (TQFTs) \cite{atiyah:88} with qubit spaces as boundaries.  This representation is Turing-complete, assuring, consistent with the universality results in \cite{deutsch:85}, all and only Turing-computable functions $f$ in Diagram \eqref{comp-scat}.  We proceed in \S \ref{uqc-tqnn} to consider TQFTs in a spin-network basis \cite{penrose:71,smollin:97}, which we have previously developed as models of quantum computation under the rubric of {\em topological quantum neural networks} (TQNNs) \cite{marciano:22, marciano:24}.  A TQNN is, effectively, a superposition of conventional QNNs; we show in \cite{marciano:22, marciano:24} that the classical limit of a TQNN that computes a function $f$ is a deep-learning network that computes $f$.  Here we show explicitly that spin-networks are universal as a data structure, then show in detail that TQNNs implement UQC, employing the work of Freedman et al. \cite{freedman:02a,freedman:02b} to establish a correspondence between UQC and the theory of 3-manifold TQFTs. To this purpose, we use the underlying topological skein theory of TQNNs as ribbon graphs to directly identify the {\em Reshetikhin-Turaev invariant} of a TQNN as a {\em Turaev-Viro quantum error correction code} implementing UQC (reviewed in \cite{koenig:10}).  This universality of quantum processes implemented within a TQNN is summarized by Theorem \ref{main-tqnn-1}. We then proceed in \S \ref{scatter-ampl}, towards a formal correspondence between UQC as implemented by TQNNs and amplituhedra, and briefly examine its implications for spin-foam models and BF theory. In \S\ref{exec-th}, this formal correspondence is stated by Theorem \ref{main-tqnn-2} whose proof draws upon the concept of the {\em execution trace} of a quantum computation, as outlined earlier in \S\ref{exec-trace}.  Theorem \ref{main-tqnn-2} establishes that TQNNs have associated amplituhedra; the association of amplituhedra with generic quantum processes then follows from the universality of TQNNs are implementations of UQC.  We conclude in \S \ref{concl} with a brief discussion of amplituhedra as representations of generic quantum processes, pointing out that edge complexity may provide a generic measure of QRF sharing in information-exchange processes, and pointing to potential applications in computational complexity, neuroscience, and other areas.

%%%%%%%%%%%%%%%%%%%%%%%%%%%%%%%%%%%%%%%%%%%

\section{Computation and scattering in an operational setting}\label{locc}

\subsection{LOCC protocols with external clocks}\label{clocks}

Let us start with some physical intuitions.  Suppose Alice ($A$) prepares a quantum state $\psi_i$ that can be regarded, in practice, as protected from decoherence by some environmental constraints.  This state evolves under the action of some unitary $U$ for some time $\Delta t$, after which Bob ($B$) measures the resulting state $\psi_f = U_{\Delta t}(\psi_i)$.  What is required for Bob's measurement of $\psi_f$ to confirm that the action of $U$ over $\Delta t$ has preserved some relevant symmetry?  Clearly, Bob needs to know $\psi_i$.  Bob cannot measure $\psi_i$, but can obtain a description of Alice's preparation procedure --- or equivalently, an outcome of Alice's measurement of $\psi_i$ --- as a classical communication from Alice.\footnote{If it is Alice who, after evolving through $\Delta t$, measures $\psi_f$, she must rely on her classical memory of her previous outcome, a ``message to her future self''.}

In this scenario, Alice and Bob each perform local operations --- state preparation and/or measurement --- on a shared quantum system, and also communicate classically.  It is, therefore, a LOCC protocol \cite{chit:14} with $U$ as the quantum channel; such protocols are comprehensively described in \cite{chit:14}, to which we refer for details.  The protocol is rendered causal by the directionality of the classical communication from Alice to Bob.  To agree on the interval $\Delta t$, Alice and Bob must also share access to a clock that is implemented by their joint environment, but is isolated from and hence not interfering with $U$.

The canonical LOCC protocol is a Bell/EPR experiment, where Alice and Bob must agree, via classical communication, to employ specified detectors in specified ways, and must later exchange their accumulated data, or transfer to the 3rd party Charlie ($C$), in the form of classical records.  Treating $A$ and $B$ as mutually-separable quantum systems and using the representations of QRFs as operators and channels as TQFTs developed in \cite{fgm:22a}, we showed in \cite{fgm:24a} that any LOCC protocol can be represented by Diagram \eqref{locc-diag}, in which $A$ and $B$ are separated from their joint environment $E$ by a holographic screen $\mathscr{B}$ \cite{fm:20,fgm:22b}, implement read/write QRFs $Q_A$ and $Q_B$, respectively, and communicate via classical and quantum channels implemented by $E$.

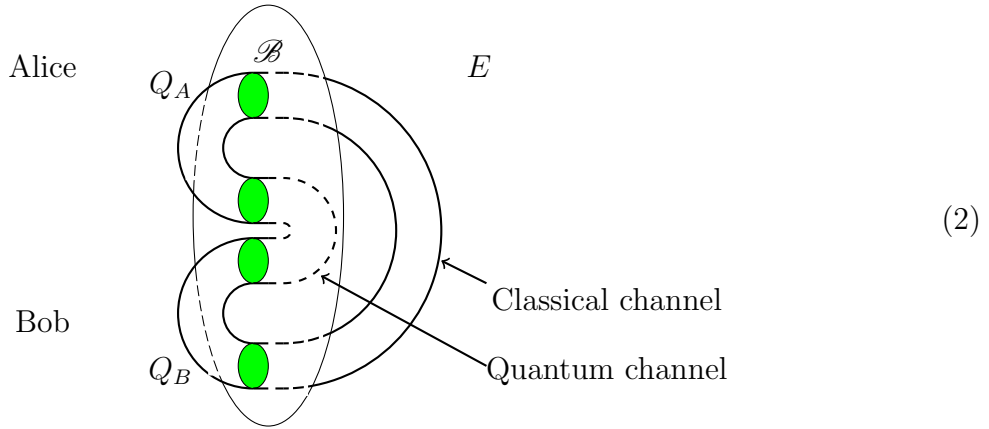
\begin{equation} \label{locc-diag}
\begin{gathered}
\begin{tikzpicture}[every tqft/.append style={transform shape}]
\draw[rotate=90] (0,0) ellipse (2.8cm and 1 cm);
\node[above] at (0,1.9) {$\mathscr{B}$};
\draw [thick] (-0.2,1.9) arc [radius=1, start angle=90, end angle= 270];
\draw [thick] (-0.2,1.3) arc [radius=0.4, start angle=90, end angle= 270];
\draw[rotate=90,fill=green,fill opacity=1] (1.6,0.2) ellipse (0.3 cm and 0.2 cm);
\draw[rotate=90,fill=green,fill opacity=1] (0.2,0.2) ellipse (0.3 cm and 0.2 cm);
\draw [thick] (-0.2,-0.3) arc [radius=1, start angle=90, end angle= 270];
\draw [thick] (-0.2,-0.9) arc [radius=0.4, start angle=90, end angle= 270];
\draw[rotate=90,fill=green,fill opacity=1] (-0.6,0.2) ellipse (0.3 cm and 0.2 cm);
\draw[rotate=90,fill=green,fill opacity=1] (-2.0,0.2) ellipse (0.3 cm and 0.2 cm);
\draw [rotate=180, thick, dashed] (-0.2,0.9) arc [radius=0.7, start angle=90, end angle= 270];
\draw [rotate=180, thick, dashed] (-0.2,0.3) arc [radius=0.1, start angle=90, end angle= 270];
\draw [thick] (-0.2,0.5) -- (0,0.5);
\draw [thick] (-0.2,-0.1) -- (0,-0.1);
\draw [thick] (-0.2,-0.9) -- (0,-0.9);
\draw [thick] (-0.2,-0.3) -- (0,-0.3);
\draw [thick, dashed] (0,0.5) -- (0.2,0.5);
\draw [thick, dashed] (0,-0.1) -- (0.2,-0.1);
\draw [thick, dashed] (0,-0.9) -- (0.2,-0.9);
\draw [thick, dashed] (0,-0.3) -- (0.2,-0.3);
\node[above] at (-3,1.7) {Alice};
\node[above] at (-3,-1.7) {Bob};
\node[above] at (2.8,1.7) {$E$};
\draw [ultra thick, white] (-0.9,1.5) -- (-0.7,1.5);
\draw [ultra thick, white] (-1,1.3) -- (-0.8,1.3);
\draw [ultra thick, white] (-1,1.1) -- (-0.8,1.1);
\draw [ultra thick, white] (-1,0.9) -- (-0.8,0.9);
\draw [ultra thick, white] (-1.1,0.7) -- (-0.8,0.7);
\draw [ultra thick, white] (-1.1,0.5) -- (-0.8,0.5);
\draw [ultra thick, white] (-1,-0.9) -- (-0.8,-0.9);
\draw [ultra thick, white] (-1,-1.1) -- (-0.8,-1.1);
\draw [ultra thick, white] (-1,-1.3) -- (-0.8,-1.3);
\draw [ultra thick, white] (-0.9,-1.5) -- (-0.7,-1.5);
\draw [ultra thick, white] (-0.9,-1.7) -- (-0.7,-1.7);
\draw [ultra thick, white] (-0.8,-1.9) -- (-0.6,-1.9);
\draw [ultra thick, white] (-0.8,-2.1) -- (-0.6,-2.1);
\node[above] at (-1.3,1.4) {$Q_A$};
\node[above] at (-1.3,-2.4) {$Q_B$};
\draw [rotate=180, thick] (-0.2,2.3) arc [radius=2.1, start angle=90, end angle= 270];
\draw [rotate=180, thick] (-0.2,1.7) arc [radius=1.5, start angle=90, end angle= 270];
\draw [thick] (-0.2,1.9) -- (0,1.9);
\draw [thick] (-0.2,1.3) -- (0,1.3);
\draw [thick, dashed] (0.2,1.9) -- (0,1.9);
\draw [thick, dashed] (0.2,1.3) -- (0,1.3);
\draw [thick] (-0.2,-1.7) -- (0,-1.7);
\draw [thick] (-0.2,-2.3) -- (0,-2.3);
\draw [thick, dashed] (0.2,-1.7) -- (0,-1.7);
\draw [thick, dashed] (0.2,-2.3) -- (0,-2.3);
\draw [ultra thick, white] (0.3,2) -- (0.3,1.2);
\draw [ultra thick, white] (0.5,2) -- (0.5,1.2);
\draw [ultra thick, white] (0.7,1.9) -- (0.7,1.1);
\draw [ultra thick, white] (0.3,-2.4) -- (0.3,-1.5);
\draw [ultra thick, white] (0.5,-2.4) -- (0.5,-1.5);
\draw [ultra thick, white] (0.7,-1.8) -- (0.7,-1.5);
\node[above] at (4.5,-2.4) {Quantum channel};
\draw [thick, ->] (2.9,-2) -- (0.7,-0.8);
\node[above] at (4.5,-1.4) {Classical channel};
\draw [thick, ->] (2.9,-0.9) -- (2.3,-0.6);
\end{tikzpicture}
\end{gathered}
\end{equation}
Note both that $A$ and $B$ being mutually separable is required for the assumption of classical communication via a causal channel in $E$, and that this assumption renders $Q_A$ and $Q_B$ noncommutative and hence subject to quantum contextuality \cite{fg:21,fg:23}.

Two defining characteristics of LOCC protocols are worth emphasizing \cite{fgmz:24}:
\begin{itemize}
    \item[(1)] $A$ and $B$ both perform only {\em local} operations.  They must, therefore, each employ spatial QRFs acting on $\mathscr{B}$, which we will denote $X_A$ and $X_B$, respectively, with respect to which they specify the position of the quantum degrees of freedom that they manipulate, e.g. the positions of the detectors in a Bell/EPR experiment.  These spatial QRFs must commute with the QRFs $Q_A$ and $Q_B$ that they, respectively, employ to manipulate the quantum channel, i.e. $[X_A, Q_A] = [X_B, Q_B] =_{def} 0$.

    \item[(2)]
$A$ and $B$ must both comprise sufficient degrees of freedom for them both to implement their respective QRFs and to communicate classically.  This is, effectively, a large $N$-limit that assures their separability as physical systems.
\end{itemize}

Canonical LOCC protocols such as Bell/EPR allow, and may require, Alice and Bob to make simultaneous measurements of their shared quantum channel.  One can, however, impose causality by requiring that Bob make his measurement only after, according to some shared clock, Alice has made hers, and moreover requiring that Alice does not interact with the channel during this interval. It is natural to represent such causal protocols by Diagram \eqref{causal-locc}, where one channel is treated as quantum and the other as classical:

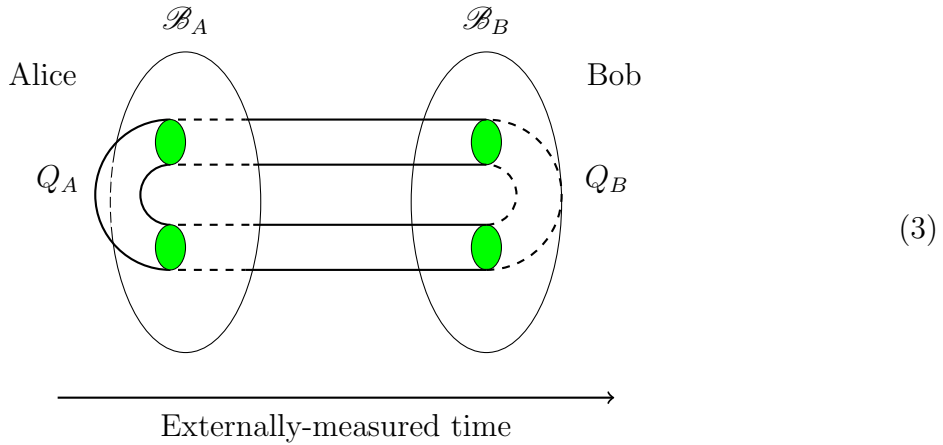
\begin{equation} \label{causal-locc}
\begin{gathered}
\begin{tikzpicture}[every tqft/.append style={transform shape}]
\draw[rotate=90] (0,0) ellipse (2cm and 1 cm);
\node[above] at (0,2.1) {$\mathscr{B}_{A}$};
\draw [thick] (-0.2,1.1) arc [radius=1, start angle=90, end angle= 270];
\draw [thick] (-0.2,0.5) arc [radius=0.4, start angle=90, end angle= 270];
\draw[rotate=90,fill=green,fill opacity=1] (0.8,0.2) ellipse (0.3 cm and 0.2 cm);
\draw[rotate=90,fill=green,fill opacity=1] (-0.6,0.2) ellipse (0.3 cm and 0.2 cm);
\draw [ultra thick, white] (-1,0.5) -- (-0.8,0.5);
\draw [ultra thick, white] (-1,0.3) -- (-0.8,0.3);
\draw [ultra thick, white] (-1.1,0.1) -- (-0.8,0.1);
\draw [ultra thick, white] (-1.1,-0.1) -- (-0.8,-0.1);
\draw [ultra thick, white] (-1,-0.3) -- (-0.8,-0.3);
\draw[rotate=90] (0,-4) ellipse (2cm and 1 cm);
\draw[rotate=90,fill=green,fill opacity=1] (0.8,-4) ellipse (0.3 cm and 0.2 cm);
\draw[rotate=90,fill=green,fill opacity=1] (-0.6,-4) ellipse (0.3 cm and 0.2 cm);
\draw [thick, dashed] (-0.1,0.5) -- (0.9,0.5);
\draw [thick, dashed] (-0.1,1.1) -- (0.8,1.1);
\draw [thick, dashed] (-0.1,-0.3) -- (0.9,-0.3);
\draw [thick, dashed] (-0.1,-0.9) -- (0.8,-0.9);
\draw [thick] (0.9,0.5) -- (4,0.5);
\draw [thick] (0.8,1.1) -- (4,1.1);
\draw [thick] (0.9,-0.3) -- (4,-0.3);
\draw [thick] (0.8,-0.9) -- (4,-0.9);
\draw [rotate=180, thick, dashed] (-4,0.9) arc [radius=1, start angle=90, end angle= 270];
\draw [rotate=180, thick, dashed] (-4,0.3) arc [radius=0.4, start angle=90, end angle= 270];
\node[above] at (4,2.1) {$\mathscr{B}_{B}$};
\node at (-1.7,0.3) {$Q_A$};
\node at (5.6,0.3) {$Q_B$};
\node at (-1.9,1.7) {Alice};
\node at (5.7,1.7) {Bob};
\draw [thick, ->] (-1.7,-2.6) -- (5.7,-2.6);
\node at (2,-3) {Externally-measured time};
\end{tikzpicture}
\end{gathered}
\end{equation}

If we interpret Alice's action with $Q_A$ as a preparation of the state of the sector of her boundary $\mathscr{B}_A$ that is input to the quantum channel, and interpret Bob's action with $Q_B$ as a measurement of the state of the corresponding sector of $\mathscr{B}_B$, Diagram \eqref{causal-locc} describes either an instance of quantum computation or scattering.  In either of these cases, Alice employs the classical channel to send Bob a description of her preparation's procedures, and hence of the state she has prepared.

\subsection{Execution traces and computations}\label{exec-trace}

Computation is often thought of as an abstract process implemented by an algorithm that is itself specified abstractly by a program written in some formalized language.  In practice, however, all computations are implemented on finite physical systems that employ finite space, time, and energy as resources.  Any instance of computing some function $f$ on some finite data $x$ using some finite device/machine $M$ is a finite sequence of states of $M$ --- an {\em execution trace} of $f$ on $M$ --- from some initial state of $M$ that encodes $x$ to some final state of $M$ that encodes $f(x)$.

This notion of an execution trace of a function $f$, and of computation of $f$ as abstraction to ``all possible'' inputs to $f$, can be made precise following \cite{horsman:14}.  We consider $M$ to be effectively isolated, and let $\mathscr{P}_M(t)$ be the time-propagator for states $\rho_M(t)$ of $M$ as in Diagram \eqref{comp-scat}, where by employing the density representation we acknowledge that the states of interest are coarse-grained to some in-practice feasible resolution. The function $f: x \mapsto f(x)$, where $x$ and $f(x)$ are finite instances of some finite data types.  The device $M$ {\em implements} $f$ on $x$ if and only if there exist maps $\zeta$ and $\xi$ such that Diagram \eqref{comp-def} commutes, i.e. $f \zeta = \xi \mathscr{P}_M(t)$ over some finite interval $\Delta t = t_f - t_i$, the {\em execution time} of $f$ on $M$ for input $x$.

% \begin{equation} \label{comp-def}
% \begin{gathered}
% \begin{tikzpicture}
% \node at (0,.4) {$f$};
% \draw [thick, ->] (-1,0) -- (1,0);
% \node at (-1.4,0) {$x$};
% \node at (1.5,0) {$f(x)$};
% \node at (-1.4,-2) {$\rho_M(t_i)$};
% \node at (1.6,-2) {$\rho_M(t_f)$};
% \draw [thick, ->] (-0.7,-2) -- (0.8,-2);
% \node at (-.1,-2.4) {$\mathscr{P}_M(\Delta_t)$};
% \draw [thick, ->] (-1.4,-1.7) -- (-1.4,-0.3);
% \node at (-1.7,-1) {$\zeta$};
% \draw [thick, ->] (1.5,-1.7) -- (1.5,-0.3);
% \node at (1.8,-1) {$\xi$};
% \end{tikzpicture}
% \end{gathered}
% \end{equation}
In what follows, we consider the space $l^\infty(X)$, where $X$ is the finite set of data types, whose states $S(l^\infty(X))$ are probability distributions on $X$, $Prob(X)$. Then, $f$ is a map $S(l^\infty(X))\longrightarrow S(l^\infty(X))$. The states $\rho_M(t)$ are linear functionals on a von  Neumann algebra $\mathscr{N}$ associated to $M$. We construct the following two categories. For the first one, We let the objects be the state spaces denoted $S(M)$ as with respect to $\mathscr{N}$, and consider as morphisms the evolution operators $S(M) \longrightarrow S(M)$. For the second category, we let the objects be $l^\infty(X)$ for some finite $X$ of data types. We let $Prob(X)$ denote the space of states on $l^\infty(X)$, and consider the morphisms $f: Prob(X) \longrightarrow Prob(X)$ defined by the pushforward of $f$.

\begin{equation}\label{comp-def}
\begin{tikzcd}
    S(l^\infty(X))\arrow[rr,"f"] & & S(l^\infty(X))\\
    & & \\
    S(M)\arrow[uu,"\zeta"]\arrow[rr,"\mathscr{P}_M(\Delta_t)"] & & S(M)\arrow[uu,swap,"\xi"].
\end{tikzcd}
\end{equation}
Therefore, commutativity at the set level can be interpreted as an endomorphism (natural transformation) of the forgetful functor from the category of convex sets with affine maps to the category of sets.

The maps $\zeta$ and $\xi$ send physical states to instances of data structures, i.e. symbols, and are therefore in a natural sense {\em interpretations} of states of $M$.  Reversing these arrows to send data structures to physical states yields {\em encoding} of data.  These maps are fundamentally stipulative, with the only constraint being diagram commutativity; hence a given physical system can, in general, be interpreted as computing multiple functions.  We can, however, constrain these maps by assuming that the device $M$ is implementing only a single procedure, and that $M$ returns to the same global ``ready'' state before receiving the next input $x$.  Effectively, these assumptions re-interpret the input $x$ as encoding both the program to be executed and the input to that program, i.e. they treat $x$ as the ``tape'' of a Universal Turing Machine \cite{turing:37}; this is indeed the first UQC model was developed in \cite{{deutsch:85}}.  With these assumptions, specifying an execution trace from some input $x$ specifies $\zeta$ and $\xi$, and therefore specifies $f$; any execution trace is, therefore, an execution trace of a unique function $f$ acting on $x$.

A computation is {\em classical at a scale} $dt$ if and only if, for some $n$ such that $n dt = \Delta t$, there are $n$ states $\rho_M(t_j)$, maps $\zeta_j$, and functions $f_j$ such that the following diagrams commute:
\begin{equation}\label{diag:subdiagrams_j}
\begin{tikzcd}
    S(l^\infty(X))\arrow[rr,"f_j"] & & S(l^\infty(X))\\
    & & \\
    S(M)\arrow[uu,"\zeta_j"]\arrow[rr,"\mathscr{P}_M(\Delta_j)"] & & S(M)\arrow[uu,swap,"\zeta_{j+1}"],
\end{tikzcd}
\end{equation}
for each $j=1, \ldots, n-1$, where $\zeta = \zeta_1$, $\xi = \zeta_n$ in Diagram \eqref{comp-def},with $S(M)$ as above. Setting $x = x_1$, $f(x) = x_n$, $\rho_M(t_i) = \rho_M(t_1)$ and $\rho_M(t_f)$, we obtain the equality $f \zeta(\rho_M(t_i)) = \xi\mathscr{P}_M(\Delta_t)(\rho_M(t_f))$ for each state $\rho_M(t_i) \in S(M)$ as in  Diagram \eqref{comp-def} by pasting all diagrams as in \eqref{diag:subdiagrams_j}. In particular, we have $f_j\cdots f_1(\zeta\rho_M(t_i)) = \zeta_{j+1}\mathscr{P}_M(t_j-t_i)(\rho_M(x_i))$.

The minimum scale $dt_{min}$ at which the process $\mathscr{P}_M(t)$ can be given a classical interpretation is its measurable decoherence scale at the resolution chosen to define the $\rho_M$.   The goal of practical quantum computation is to make $dt_{min}$ large, and hence the number $n$ of classically-characterizable ``steps'' small. Related are several programs outlined in \cite{horsman:14} for implementing such computation. In practice, purpose-built quantum computers achieve this with cryogenics and other isolation mechanisms; accelerators achieve the same result by operating at the GeV scale and above.

\subsection{Defining $Q_C$ as a QRF}

Comparing Diagrams \eqref{causal-locc} and \eqref{comp-def}, and following \cite{horsman:14} in noting that ``interpretation'' and ``measurement'' are operationally the same concept, it is clear that the maps $\xi$ and $\zeta$ in \eqref{comp-def} are realized, in any operational setting, by QRFs.  It is also clear that in practice, the description of a physical process $\mathcal{P}_M(t)$ as implementing a computation of some function $f$ is coherent only if $\xi$ and $\zeta$ are realized by the {\em same} QRF, i.e. by $Q_C$ in Diagram \eqref{comp-scat}.

In classical physics, reference frames are abstract coordinate systems, associated in practice with rulers, clocks, standard charges and masses, etc.  With the recognition that measurement entails interaction, and hence the passage to quantum theory, reference frames explicitly become physical systems, QRFs \cite{bartlett:07}.  Completely specifying a measurement operator requires specifying the associated QRF, so it is natural to represent the QRF itself as an operator.  The simplest QRF is then the $\varphi$-spin operator $(\varphi, \sigma_{\varphi})$ acting on a qubit $q_i$ to yield as output an eigenvalue in $\{+1, -1\}$.  In practice, $\varphi$ may be chosen macroscopically, e.g. by orienting the axis of a Stern-Gerlach apparatus with the Earth's gravitational field, then rendered implicit by writing the QRF as $\sigma_z$.

This can be made more explicit following \cite{mikusch:21}. Consider a measurable (value) space $\Sigma_Q$, and a locally compact group $G$ such that $\Sigma_Q$ is a homogeneous (left) $G$-space, a QRF $Q$ can be defined via a `system of covariance' $(\rm{U}_Q, \mathsf{E}_Q, \mathcal{H}_Q)$, where $\rm{U}_Q$ is a unitary representation of $G$ on a Hilbert space $\mathcal{H}_Q$, $\mathsf{E}_Q$ is a covariant {\em positive operator-valued measure} (POVM) on $\Sigma_Q$ with values in
the bounded linear operators $\mathcal{L}(\mathcal{H}_Q)$. Once these arguments are understood with $G$ given, $\mathsf{E}_Q$ is interpretable as a `frame observable' \cite{carette:25}, see also \cite{palmer:14}.  In a relational setting, these give an operational meaning to spin \cite{giaco:19a,giaco:19b}, and correspond to general spin systems such as those described by spin-networks \cite{mikusch:21}.

Generally, a transformation of QRFs is given by a unitary transformation $\rm{U}: \mathcal{H}_{Q_1} \rightarrow \mathcal{H}_{Q_2}$, where $Q_1$ and $Q_2$ are said to be {\em unitarily equivalent} if, with respect to the value space $\Sigma$, for any $X \in \mathcal{L}(\Sigma)$, we have in terms of POVMs, $\mathsf{E}_2(X) = \rm{U} \mathsf{E}_1(X) \rm{U}^*$  \cite{carette:25}; cf. \cite{mikusch:21,palmer:14,giaco:19a,giaco:19b, hamette:20}. Alternatively, the classical means of matrix transforming reference frames (coordinate systems) with respect to classical $\rm{SO}(3)$ spin transformations, can be quantized by generalizing the Euler angles as operator-valued on the relevant Hilbert spaces \cite{mikusch:21}.

\subsection{Turing-completeness of QRFs}\label{ccd-1}

If Diagram \eqref{comp-def}, and hence Diagram \eqref{comp-scat}, are to represent the computation of any Turing-computable \cite{turing:37} function $f$, then the QRF $Q_C$ must not only implement a Turing-computable function, but also act as an effective (up to memory resources) Turing-complete computational process.  We have shown this to be the case by constructing QRFs as cone-cocone (limit/colimit) diagrams (CCCDs) over finite sets of Barwise-Seligman {\em classifiers} $\mathcal{A}_i$ and their associated {\em infomorphisms} $\{f_{\alpha}, g_{\alpha \beta}\}$ \cite{barwise:97} as shown in Diagram \eqref{cccd-2} below \cite[Th 1]{fgm:22a}; see \cite{fg:21,fg:19} for relevant category-theoretic background.  The classifiers are operators assigning {\em tokens} to {\em types} in a formal language representing probabilities; in the simplest case, they are Boolean gates.\footnote{An equivalent category to that of classifiers (objects) and their infomorphisms (arrows) in \cite{barwise:97} is that of Chu spaces and their morphisms \cite{pratt:99a,pratt:99b} which collectively comprise a Turing-complete syntax, so any computation can be expressed in this way.}  A CCCD implements a hierarchical, memory-read/write computation on whatever data structure is accessed by the $\mathcal{A}_i$; in \cite{fgm:22a}, these are $(\varphi, \sigma_{\varphi})$ QRFs acting on a qubit array, i.e. a holographic screen, with the reference angles $\varphi$ chosen independently for each qubit.  They can be visualized in the form of variational autoencoders \cite{kingma:19}, as in Diagram \eqref{cccd-2}:
\begin{equation}\label{cccd-2}
\begin{gathered}
\xymatrix@C=6pc{\mathcal{A}_1 \ar[r]_{g_{12}}^{g_{21}} & \ar[l] \mathcal{A}_2 \ar[r]_{g_{23}}^{g_{32}} & \ar[l] \ldots ~\mathcal{A}_k \\
&\mathbf{C^\prime} \ar[ul]^{h_1} \ar[u]^{h_2} \ar[ur]_{h_k}& \\
\mathcal{A}_1 \ar[ur]^{f_1} \ar[r]_{g_{12}}^{g_{21}} & \ar[l] \mathcal{A}_2 \ar[u]_{f_2} \ar[r]_{g_{23}}^{g_{32}} & \ar[l] \ldots ~\mathcal{A}_k \ar[ul]_{f_k}
}
\end{gathered}
\end{equation}
In this representation, a multi-component, hierarchical QRF implements a bidirectional mapping between bit strings representing the outcomes obtained by the $\mathcal{A}_i$ --- e.g. $(\varphi, \sigma_{\varphi})$ operators acting on qubits --- and bit strings representing observational outcome as fully processed through the QRF. From this point of view, a QRF computes a function pair $\overrightarrow{g},\overleftarrow{g}$: bit strings $\leftrightarrows$ bit strings, where $\overrightarrow{g}$ and $\overleftarrow{g}$ are the composite maps in to and out from the core classifier $\mathbf{C^\prime}$ in Diagram \eqref{cccd-2}.

\subsection{Implementing computation with scattering models}
\label{implement}

Consider now an operational setting in which $A$ deploys a QRF $Q$ acting on a sector $\Sigma$ of the boundary $\mathscr{B}$ between $A$ and some finite system $U$.  The action of $Q$ on $\Sigma$ yields an outcome, a finite bit string.  Suppose this outcome specifies the initial state of some scattering process.

Now suppose that $U$, which is unobserved by $A$ except in the sector $\Sigma$, implements a propagator $\mathcal{P}_U$
that restricts to the identity on the boundary $\mathscr{B}$ outside of $\Sigma$, and can be characterized by operator-valued functions $\mathcal{F}$, $\mathcal{G}$ with respect to the Hilbert space $\mathcal{H}_{\Sigma}$, such that $\Sigma_1 = \mathcal{F}(\Sigma)$, and $\Sigma_2 = \mathcal{G}(\Sigma)$ are sectors of a boundary $\mathscr{B}^{\prime}$ accessed by $B$ as shown in Diagram \eqref{comp-scatter}.

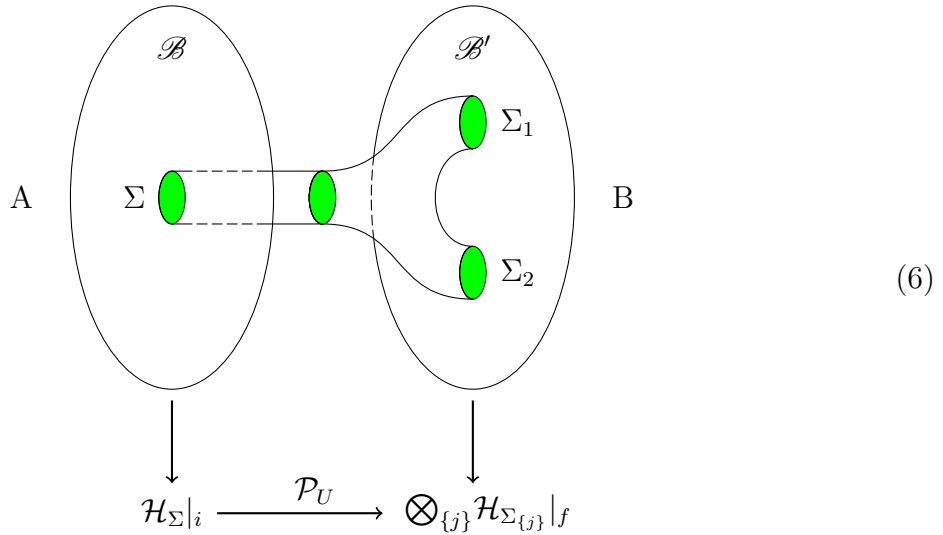
\begin{equation} \label{comp-scatter}
\begin{gathered}
\begin{tikzpicture}[every tqft/.append style={transform shape}]
\draw[rotate=90] (0,0) ellipse (2.55cm and 1.35 cm);
\node[above] at (0,1.7) {$\mathscr{B}$};
\node at (-0.5,0) {$\Sigma$};
\begin{scope}[tqft/every boundary component/.style={draw,fill=green,fill opacity=1}]
\begin{scope}[tqft/cobordism/.style={draw}]
\begin{scope}[rotate=90]
\pic[tqft/cylinder, name=a];
\pic[tqft/pair of pants, anchor=incoming boundary 1, name=b, at=(a-outgoing boundary 1)];
\end{scope}
\end{scope}
\end{scope}
\draw[rotate=90] (0,-4) ellipse (2.55cm and 1.35 cm);
\node[above] at (4,1.7) {$\mathscr{B}^{\prime}$};
\node at (2,0.7) {};
\node at (4.6,1) {$\Sigma_1$};
\node at (4.6,-1) {$\Sigma_2$};
\draw [thick, ->] (0,-2.7) -- (0,-3.8);
\draw [thick, ->] (4,-2.7) -- (4,-3.8);
\node at (0,-4.2) {$\mathcal{H}_{\Sigma}|_i$};
\node at (4.2,-4.2) {$ {\bigotimes}_{ \{ j \} } \mathcal{H}_{{\Sigma}_{ \{ j \} } }|_f$};
\draw [thick, ->] (0.6,-4.2) -- (2.8,-4.2);
\node at (1.9,-3.9) {$\mathcal{P}_U$};
\node at (-2,0) {A};
\node at (6,0) {B};
\draw [ultra thick, white] (0.3,0.5) -- (0.3,-0.5);
\draw [ultra thick, white] (0.5,0.5) -- (0.5,-0.5);
\draw [ultra thick, white] (0.7,0.5) -- (0.7,-0.5);
\draw [ultra thick, white] (0.9,0.5) -- (0.9,-0.5);
\draw [ultra thick, white] (1.1,0.5) -- (1.1,-0.5);
\draw [ultra thick, white] (2.6,0.4) -- (2.8,0.4);
\draw [ultra thick, white] (2.6,0.2) -- (2.8,0.2);
\draw [ultra thick, white] (2.6,0.0) -- (2.8,0.0);
\draw [ultra thick, white] (2.6,-0.2) -- (2.8,-0.2);
\draw [ultra thick, white] (2.6,-0.4) -- (2.8,-0.4);
\end{tikzpicture}
\end{gathered}
\end{equation}

If $B$ receives $A$'s outcome measurement as a classical communication, Diagram \eqref{comp-scatter} depicts a causal LOCC protocol as in Diagram \eqref{causal-locc}.  With no further information or constraints, $B$ can interpret observational outcomes obtained from $\Sigma_1$ and $\Sigma_2$ as experimental outcomes from a scattering process or as computed outcomes from a simulation, i.e. $B$ can interpret the outcome data using either a scattering model (a QRF $Q_S$) or a computational model (a QRF $Q_C$).  The sectors $\Sigma_1$ and $\Sigma_2$ are, however, cobordant and therefore unitarily related.  The map $Imp$ that relates these models in Diagram \eqref{comp-scat} is, therefore, just $Q_S Q_C^{\dagger}$.

%%%%%%%%%%%%%%%%%%%%%%%%%%%%%%%%%%%%%%%%%%%%%%

\section{Spin-networks and the TQNN implementation of UQC}\label{uqc-tqnn}

We can summarize the previous section in a simple slogan: operationally, quantum computation is unitary evolution interpreted as computation.  This has been understood at least since \cite{deutsch:85}.  Models of quantum computation are ways of making this interpretation obvious.  We focus here on TQNNs --- TQFTs in a spin-network basis \cite{marciano:22, marciano:24} --- as a model of quantum computation that anticipates the structure of amplituhedra.  From an intuitive perspective, it is obvious that TQNNs implement UQC.  Any unitary evolution / quantum computation can be expressed as a TQFT, so any unitary evolution / quantum computation can be implemented by a TQNN. Showing this explicitly reveals the structural details that make this universality useful.\footnote{Note that our framework does not identify a TQFT on spin networks {\it per se} as a neural network. Rather, TQNNs are defined as computational schemes in which transition amplitudes, computed via TQFT techniques on spin network structures, play the role of learnable maps. This distinction is explicitly developed in our previous work where we also showed that the semi-classical limit of TQNNs reproduces classical neural network architectures \cite{marciano:22,marciano:24}. }

\subsection{Spin-networks as a universal data structure}\label{spin-comp-1}

It is also obvious, intuitively, that spin-networks provide a universal data structure: ordered strings of bits provide a universal data structure for classical computation in general \cite{turing:37}, and spin-networks are a quantum generalization of bit strings that allow for generic relations between data items.  The idea of a spin network as a combinatorial method (which is intrinsically computational --- see below) describing a discrete model for a quantum theoretical structure underlying some classical space-time manifold was introduced by R. Penrose in \cite{penrose:71}.  In its basic form, a spin network is an $n$-valent graph in which each edge is labeled by a spin-$j$ representation of the angular momentum covering group  $\rm{SU}(2)$, where $j=0, \frac{1}{2}, 1, \ldots$, satisfying $\vert j_1 - j_2 \vert \leq j_3 \leq j_1 + j_2$, and $j_1 + j_2 + j_3 \in \mathbb{Z}$. Specifically, and more generally, let $G$ be a connected compact Lie group, and let $\Phi$ be a finite, directed graph (FDG) consisting of $N_E$ edges, and $N_V$ vertices, where $E,V$, respectively, denote the edge and vertex sets of $\Phi$. Spin network states are prescribed in \cite{baez:96} as:
\begin{itemize}
\item[i)] $\{ j_i: i \in 1, \ldots, N_E \}$ irreducible $G$-representations assigned to each edge;

\item[ii)] $\{ v_k: k \in 1, \ldots, N_V \}$ intertwining operators --- i.e. $G$-equivariant maps between representations in i) --- assigned to each vertex;

\end{itemize}
with the above conditions satisfied; see e.g. \cite{vaid:22} for the $\rm{SU}(2)$ study-case.

Following \cite{baez:96}, a principal $G$-bundle is definable on $V$, entailing the definition of a space $\mathcal{A}$ of connections on $\Phi$, for which the latter admits a group $\mathcal{G}$ of gauge transformations. A key result is that $L^2(\mathcal{A}/\mathcal{G})$ is spanned by the aforementioned spin-network states. Replacing $V$ by a real analytic manifold, the spin-network $\Phi$ becomes embedded in $M$ as a FDG, and again, $L^2(\mathcal{A}/\mathcal{G})$ is spanned by the spin-network states\footnote{It is well-known, as pointed out in \cite{basu:70,cooper:66}, that a FDG can be viewed as an abstract representation of the flow of control through a computational program. Effectively, if the program instructions are represented by vertices (e.g. intertwiners) of the FDG, then the edges (e.g. spin states and their representations) specify conditions under which execution of the program conform to that instruction. In particular, this amounts to a spin network as an abstract representation of the flow of control within a computational procedure. If $\Phi_{\mathbf{A}}$ and $\Phi_{{B}}$ denote respective spin-networks of states of systems $\mathbf{A}$ and $\mathbf{B}$, then their transformations are those of FDGs. The latter transformations follow by well-known standard rules, as typically implemented by certain moves: e.g. `safe vertex/edge deletion', `stretching', etc. \cite{basu:70,cooper:66}. }  (specifically, \cite[Th. 1,~Th. 2]{baez:96}).

For latter purposes (as below) we could also view a (classical) spin-network as consisting of a {\em cubic ribbon graph} $\Gamma$ --- i.e. an abstract trivalent graph with cyclic ordering of its edges at each vertex --- together with a coloring of its edges set by $\mathbb{N}$, the natural numbers; see e.g. \cite{garou:13}.

\subsection*{Example: Linear optics for photons and computation - spin-networks } \label{optics}

The proposal of Knill-Laflamme-Milburn \cite{knill:01} outlines a convenient and efficient means for quantum computation and general quantum information processing (QIP), using the technology of linear optics applied to photon interference. It involves beam splitter, phase shifters, and other methods --- see e.g. \cite{garcia:10, reck:94, svozil:96} for a discussion. It starts with the fundamental observation that a qubit can be realized by one photon in two possible optical modes --- e.g. via horizontal or vertical polarization. Through several steps, this in turn leads to showing that non-deterministic quantum computation is possible by applying methods of linear optics. Effectively, using entangled states and quantum teleportation to ensure the effective gating, together with quantum coding, culminates in the efficiency of {\em linear optics quantum computation} (LOQC), a means of training towards fault-tolerance for photon loss, detector inefficiency and phase decoherence.

The basic element of an arbitrary spin-network is the coupling between spins as described by the XY Hamiltonian \cite{yang:06}, i.e.
\begin{equation}\label{xy-ham}
H_{XY} = \sum_{\langle i j\rangle} (J_{ij} \sigma_{i}^{+} \sigma_{j}^{-} + h.c.) \,,
\end{equation}

where $\sigma_{i}^{+} \sigma_{j}^{-}$ are the Pauli spin operators, and $J_{ij}$ the coupling coefficients, and $h.c.$ indicates Hermitian conjugate --- the expression in \eqref{xy-ham} for an XY Hamiltonian is explicitly unraveled in \cite[\S 2.2]{grimaudo:22}.

At the same time, beam splitters are fundamental tools of quantum (and classical) optics. From \eqref{xy-ham}, \cite{yang:06} develops a theory for spin wave propagation in {\em star-shaped spin-networks}, a basic form of which is a Y-shaped spin-network that can be implemented in QIP.

Given a single photon in $2^n$ input spatial modes, the scattering matrix corresponds to the unitary operator that determines the quantum evolution of a system of $n$ qubits. Accordingly, a $n$-qubit quantum computer needs $2^n$ orthogonal states in order to encode all the possible values of $n$-bits, along with superposition states. Using encoding on {\em orbital angular momentum} (OAM) states of light, \cite{garcia:10} engages the following techniques: i) phase shifters; ii) holograms; iii) beam splitters; and iv) OAM filters.

Using the appropriate gating methods, \cite{garcia:10} demonstrates that a single photon can be optically manipulated to implement any quantum computation, i.e. to execute, for any computable function $f$ and input $x$ on which $f$ is defined, a superposition of paths through the relevant Hilbert space interpretable as an execution trace for $f$ acting on $x$.  How the latter can involve entanglement between stationary qubits based on resonance scattering in the framework of spin-networks is treated in \cite{jin:09}.

\subsection{Phase gating in computational processes and phase of scattering}\label{gp-scatter}

Since the initial construction of a universal quantum Turing Machine by Deutsch \cite{deutsch:85}, a number of UQC models have been developed, with {\em quantum circuit models} (QCMs) \cite{nielsen:00} the most well-suited to currently-available technology.  The principle class of interest is the class of functions that instantiate quantum computation characterized by a bounded-error and achieved in a polynomial time (BQP) --- in computational complexity theory, BQP denotes indeed the class of decision problems solvable by a quantum computer in polynomial time, with an error probability of at most 1/3 for all instances. The fundamental operation in a QCM is a {\em gate} $\mathfrak{G}$, commonly defined as a unitary transformation on either $\mathbb{C}_i^2$, or $\mathbb{C}_i^2 \otimes \mathbb{C}_j^2$, $1 \leq i,j \leq n$, with the identity defined on all remaining factors. A QCM then comprises a sequence $\{ \mathfrak{G}_k \}$ of $K$ gates ($1 \leq k \leq K$), for some $K$, as applied to some tensor product $(\mathbb{C}^2)^{\otimes k(n)}$ of $k$-qubits, for some $n$ \cite{deutsch:89}.  The sequence of gates specifies the ``program'' being implemented and hence the function $f$ being computed; hence it determines the path through Hilbert space, i.e. the execution trace, for any input $x$.We note that such circuit models may involve (pseudo)random quantum states in terms of random unitary transformations induced by entanglement, comprising an essential resource for UQC, where derandomization is achieved via implementing suitable algorithms, as in  e.g. \cite{turner:16, karayel:25}.

The QCM also provides an intuitive way to see a scattering process as a computational procedure, with ``in'' states and ``out'' states corresponding to the ``input'' and ``output'' states, respectively, of a given quantum circuit \cite[\S 4.3]{vaid:22}. In these cases, collision can be viewed, informally, as ``gating''. From a scattering perspective in terms of phase we can be more precise as follows. Spin chains (say of a QCM model) are derivable from spin networks such as those incorporated into our TQNNs \cite{mar:20}. Reference \cite{kay:05} explicitly demonstrates how geometric phase (GP), including the well--known Berry phase (BP) \cite{berry:90}, can be induced from these networks. Important for our purposes is that BP/GP is highly instrumental for optimal quantum gating, so resulting in fault-tolerance and noise reduction, and thus essential for QIP and QECC (see e.g. \cite{zanardi:99,sjoqvist:08,sjoqvist:15,otten:21,chen:24,
casagrande:24}). From another direction, GP can be seen as ubiquitous in scattering processes and their evolution (see e.g. \cite{liu:11,wang:24,li:25}).  These phases can be related in terms of angular measurement, by implementing standard geometric transformations such as rotation, translation, dilation, conformal transformation, and compositions of these:
$$
\rm{GP~for~(quantum)~gating} ~ \overset{\rm{geometric~ transformation}}\Longleftrightarrow~ \rm{GP~for~scattering}
$$
For the more general topological phases a transformation between these can in principle be achieved by means of a suitable homotopy of the paths along which the respective phases are defined  (see e.g. \cite{aguilar:20}).

\subsection{Universality of topological quantum computers}\label{q-circuits}

Freedman et al. \cite{freedman:02a,freedman:02b} studied computations of BQP functions, and proved universality, for the model of topological quantum computation originally proposed by Kitaev \cite{kitaev:77}.  This model employs low dimensional topology and the theory of anyons, the motion of which in a 2D system defines a braid in 2+1 dimensions.  A QCM can be associated with a spin-network $\Phi$ by commencing from a tensor product $\mathbb{C}_1^2 \otimes \cdots \otimes \mathbb{C}_n^2$ of $n$-copies of $\mathbb{C}^2$ --- effectively, we use qubits to model a system of $n$ non-interacting spin-$\frac{1}{2}$ particles.  From the scattering perspective, an intertwiner is basically a quantum gate providing a unitary transformation acting on a set of incoming colored edges that, in turn, produces a result in terms of outgoing colored edges --- see e.g. \cite{vaid:22}.

Additionally, the spin-color $j_k=k/2$ assigned to a generic edge of the spin-network $\Phi$ can be thought to result from symmetrizing (by means of the Wenzl-Jones projector) $k$ fundamental --- i.e. with spin $j=1/2$ --- representations. While the `action' of a trivalent intertwiner follows from the compatibility conditions among representations flowing into a node/gate. In other words, it expresses the conservation of the total number of incoming and outgoing spin-$1/2$ representations, symmetrized along each edge through the Wenzl-Jones projector. The total number of fundamental representations symmetrized within a spin-network then represent a `system' of $n$ non-interacting spin-$\frac{1}{2}$ particles, with representation space provided by the tensor product $\mathbb{C}_1^2 \otimes \cdots \otimes \mathbb{C}_n^2$ of $n$-copies of $\mathbb{C}^2$.

It is useful here to recall some facts relating to how TQFTs can be defined on triangulated 3-manifolds. This is largely due to Atiyah's ground-breaking work on the subject \cite{atiyah:88,atiyah:90}. Initially, we can think of a partition function that is the sum over the states within the interior of the manifold. Ponzano and Regge \cite{ponz:68} defined a spin state by a labeling, or coloring, of each edge of the triangulation by half-integer representations $j$, as above. This idea was taken some steps further by Turaev and Viro (TV) \cite{tv:92} on taking a triangulation of a 3-manifold $N$, on which a state, or coloring, is given by labeling each edge of the triangulation by a representation of the quantum group ${\rm{U}}_q({\rm sl}(2))$, where $q$ is a $2r^{\rm th}$-root of unity for $r\geq 3$. Here, the partition function $Z_{\rm{TV}}(N)$ for $N$ is given by a finite sum over states within the interior  of $N$, of the product of $q-6j$ symbols $\{ 6j^{(t)}\}_q$ corresponding to tetrahedra $t$, weighted by specific function.  Following \cite{foxon:95}, we later provide the expression for $Z_{\rm{TV}}(N)$, and also describe its functorial nature, both of which will be needed to formulate its relation to UQC, which in turn will be used to derive the universality of TQNN-based computation.

In the limit as the {\em deformation parameter} $q \lra 1$, the TV partition function reduces to that of \cite{ponz:68}, in which the $j_i$ label representations of $\rm{SU}(2)$. It is by these means that spin-networks can be used to study the relationship of the model of \cite{ponz:68} to the loop representation of quantum GR \cite{rovelli:93}. In particular, as pointed in \cite{foxon:95}, the dual graph to a colored triangulation of \cite{ponz:68} is interpreted as a spin-network in a 1-1 correspondence.

%%%%%%%%%%%%%%%%%%%%%%%%%%%%%%%%%%%%%%%%%%%%%%%%%%%%%%%%%%%%%%%%%%%%%%%%%%%%%%%%%%%%%%%%%%%%%

\subsection{Relation between Reshetikhin-Turaev model and Turaev-Viro models} \label{rt-tv}

A famous result of Walker \cite{walker} and Turaev \cite{turaev_shadows} has determined the relation between the Reshetikhin-Turaev (RT) and Turaev-Viro (TV) $3$-manifold invariants. Although their construction relies on fairly different approaches, it turns out that they are closely related. Roughly speaking, we have that one is the square root of the other. More precisely, one finds that the TV invariant is the absolute value squared of the RT invariant.  This relationship is formally stated in terms of spherical  categories \cite{barrett:99} as given in \eqref{wrt-tv} within \S\ref{rt-inv}.

We first briefly recall the construction of the invariants, and then consider Roberts' proof of their equivalence in \cite{roberts:95}, based on the chain-mail invariants. The formulation of the correspondence via chain-mail is employed here because it is particularly suitable to two fundamental aspects of this article: relating topological quantum computing to TQNNs, and relating the latter to the amplituhedron.

We briefly summarize the definitions of RT and TV models here.  Detailed definitions can be found in \cite{turaevquantum,Turaev:1992hq,rt:91}.

The RT model can be defined over any modular tensor category $\mathcal V$, with a choice of a finite family of objects $\{V_i\}_{i\in I}$ satisfying some compatibility conditions. Namely, one needs the objects to satisfy the normalization, duality and non-degeneracy axioms, along with the fact that every object in $\mathcal V$ is dominated by the finite family. Then, given a closed $3$-manifold denoted by $\sfM$, and a framed link $L_{\sfM}$ which presents $\sfM$ via surgery, one can define the RT invariant as follows. By color of $L_{\sfM}$, we define any assignment of an object from $\{V_i\}_{i\in I}$. To each coloring we can assign a link invariant which is the operator invariant obtained by replacing the crossings by braiding in the category, and replacing maxima and minima by (co)pairings induced by duality in $\mathcal V$. A linear combination of these elements of the ground ring, by means of the quantum dimensions of each coloring, is a manifold invariant which we denote by $\tau_{\rm RT}(\sfM)$ after a suitable normalization that depends on the signature of $L_{\sfM}$. This construction can be generalized to manifolds with boundary, where we obtain a functor to the category of vector spaces, i.e. a TQFT. The functor associated to this construction will be indicated by $\tau_{\rm RT}$ as well, as the two constructions coincide when the functor is applied to the $1$-dimensional ground ring for closed manifolds.

The TV model follows a completely different approach, and it is constructed via triangulations of a closed $3$-manifold $\sfM$. One considers the dual triangulation associated to a triangulation $T$, as already mentioned above. The dual triangulation is considered as the geometric support for spin-networks whose admissible spin-colors, in the sense of \cite{KL}, are taken into account. Then the invariant $\tau_{\rm TV}(\sfM)$ is the sum over all admissible spin-colors of the evaluations of the products of colored tetrahedra and $\theta$-nets  spin-networks with some normalization.

It is well known \cite{turaevquantum} that $\tau_{\rm TV}(\sfM) = |\tau_{\rm RT}(\sfM)|^2$ for a closed 3-manifold $\sfM$, where we indicated by $\tau_{\rm TV}$ and $\tau_{\rm RT}$ the functors for TV and RT models, respectively.

More generally, the relation between TV and RT models is shown in \cite{turaevquantum} as TQFTs, i.e. as functors from the cobordism category to the modular (braided) category. One has the equality
\begin{equation}\label{eqn:TV_RT_functor}
    \tau_{\rm TV}\sf(M) = \tau_{\rm RT}(\sfM)\otimes \tau_{\rm RT}(\bar{\sfM})
\end{equation}
which demonstrates the relationship between the two invariants.\footnote{Such a modular category can be constructed from a Hecke algebra at roots of unity, where \cite{blanchet:00} shows explicitly how the RT invariant is recovered from that category.}

\subsection{The chain mail invariant}\label{chain-mail}

We now proceed to discuss how the chain-mail invariant is defined (see Definition 1 below), and how it is used to show the relation between RT and TV invariants. The starting point is a handle decomposition of a given $3$-manifold $\sfM$, with given $d_i$ $i$-handles, for $i=0,1,2,3$. The fundamental piece of information in a handlebody decomposition is how the $2$-handles are attached to the handlebody obtained by attaching $1$-handles to $0$-handles. This construction is closely related to the notion of a Heegaard splitting of a $3$-manifold, i.e. a presentation of $\sfM$ as $\sfM = H_1\cup_S H_2$, where $S$ is a surface along which two homeomorphic handlebodies $H_1$ and $H_2$ are glued. The existence of a Heegaard splitting is obtained by considering a triangulation of $\sfM$, whose existence is known to be guaranteed. By considering the $1$-skeleton $K^1$ of the triangulation and the dual diagram $\Gamma$, one obtains a pair of graphs that can be enlarged to a tubular neighborhood in order to obtain two complementary handlebodies. Glueing them along their boundaries, we obtain a splitting of $\sfM$. This process is schematically shown in Figure 1 and Figure 2.

%Figure~\ref{diag:handle_decomp} and Figure~\ref{diag:fat_handlebody}.

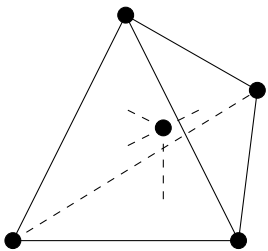
\begin{figure}[H] \label{diag:handle_decomp}
\begin{center}
\begin{tikzpicture}[every tqft/.append style={transform shape}]
\draw (0,0) -- (1.5,3) -- (3,0) -- (0,0);
\draw[dashed] (0,0) -- (3.25,2);
\draw (3.25,2) -- (1.5,3);
\draw (3,0) -- (3.25,2);
\filldraw (0,0) circle (3pt);
\filldraw (1.5,3) circle (3pt);
\filldraw (3,0) circle (3pt);
\filldraw (3.25,2) circle (3pt);
\filldraw (2,1.5) circle (3pt);
\draw[dashed] (2,1.5) -- (2,0.5);
\draw[dashed] (2,1.5) -- (1.5,1.75);
\draw[dashed] (2,1.5) -- (2.5,1.75);
\draw[dashed] (2,1.5) -- (1.5,1.25);
\end{tikzpicture}
\end{center}
\caption{Process to obtain a handle decomposition of a $3$-manifold where a tetrahedron and its dual graph are the spines of the two complementary handlebodies}
\end{figure}

\begin{figure}[H] \label{diag:fat_handlebody}
\begin{center}
\begin{tikzpicture}[every tqft/.append style={transform shape},scale=3]
\node (a) at (1,1.25) {$\mapsto$};
\filldraw (0,1.5) circle (3pt);
\draw[dashed] (0,1.5) -- (0,0.5);
\draw[dashed] (0,1.5) -- (-0.5,1.75);
\draw[dashed] (0,1.5) -- (0.5,1.75);
\draw[dashed] (0,1.5) -- (-0.5,1.25);
\draw[rounded corners] (2.3,1.5) ..controls(2.375,1.65).. (2.5,1.7);
\draw[rounded corners,dashed] (2.3,1.5) ..controls(2.4,1.5).. (2.5,1.7);
\draw (1.75,1.25) ellipse [x radius=0.075cm, y radius=0.125cm];
\draw[rounded corners] (2.3,1.5) ..controls(2.5,1.4).. (1.75,1.375);
\draw[rounded corners] (2.5,0.5) ..controls(2.5,1.25).. (1.75,1.125);
\draw[rounded corners] (2.5,0.5) ..controls(2.6,0.45).. (2.7,0.5);
\draw[rounded corners,dashed] (2.5,0.5) ..controls(2.6,0.55).. (2.7,0.5);
\draw[rounded corners] (3.1,1.225) ..controls(2.7,1.225).. (2.7,0.5);
\draw[rounded corners] (3.1,1.225) ..controls(3.135,1.325).. (3.1,1.4);
\draw[rounded corners] (3.1,1.225) ..controls(3.075,1.325).. (3.1,1.4);
\draw[rounded corners] (2.5,1.7) ..controls(2.75,1.4).. (3.1,1.4);
\end{tikzpicture}
\end{center}
\caption{The spine dual to a tetrahedron is fattened to a handlebody}
\end{figure}
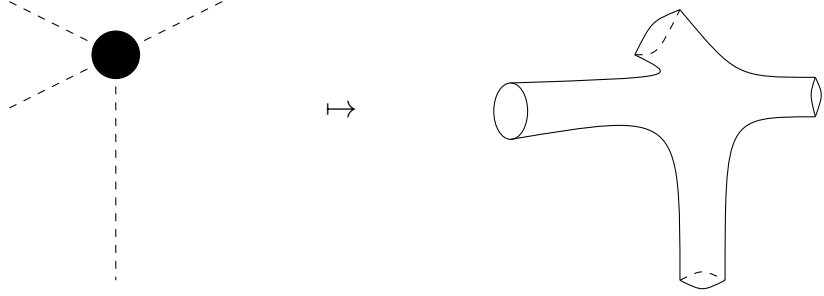

Given $\sfM$ and a handle decomposition, we associate a chain-mail link to it as follows. Consider first the attaching curves $\epsilon_i$ of the $2$-handles on the handlebody $H_1$ obtained by attaching $1$-handles to $0$-handles. We push the curves $\epsilon_i$ inside a ``small'' collar neighborhood of $S = \partial H_1$. Next, we mark the meridional curves $\delta_i$ of the $1$-handles on $H_1$ (see Figure 3). In this way, we obtain a link drawn on the handlebody $H_1$. This is defined as the chain-mail link $C(\sfM,D) \subset H$, where $D$ describes the dependence of this link on the initial data of the handle decomposition. Given an arbitrary orientation-preserving embedding of the handlebody $H$ in $S^3$, we obtain an induced link $C(\sfM,D,E)$ inside $S^3$. A numerical value is associated to $C(\sfM,D)$ by means of skein theory, a generalization of the theory of knots and links of closed 3-manifolds, following the work of Lickorish \cite{lickorish1,lickorish2,lickorish4} and Yokota \cite{yokota:97}. In particular, from the study of framed links in closed 3-manifolds, the skein module $\mathcal{S}$ of $S^3 \cong \rm{SU}(2)$ is shown to be linear 1-dimensional; thus, $\mathcal{S} S^3 = \mathbb{C}$ \cite{yokota:97} and see also \cite[\S 2-3]{lickorish4}.

\begin{definition}
{\rm
            Let $\sfM$ be a manifold with a given handlebody decomposition $D$ where $H$ is the union of its $0$ and $1$-handles, and let $E$ an orientation-preserving embedding of $H$ in $S^3$. Then, the {\em chain-mail invariant} of $\sfM$ is the element of the skein module of $S^3$ defined as
            \begin{equation}\label{eqn:chain-mail}
                    CH(\sfM) = \eta^{d_0+d_3} \Omega C(\sfM,D,E) \in \mathcal S S^3 = \mathbb C,
            \end{equation}
            %\textcolor{red}{[AM: I missed the definition of $\mathcal S $. Why is $\mathcal S S^3$ a complex number? ]\\} \textcolor{blue}{[JFG: see above Definition 1.]}

            where $d_0$ and $d_3$ are the number of $0$ and $3$-handles in the decomposition $D$, $\Omega C(\sfM,D,E)$ indicates the insertion of the $\Omega$ element of skein-theory in each component of the link $C(\sfM,D,E)$ previously defined, and $\eta$ is a normalization constant defined as $\eta = \frac{A^2-A^{-2}}{i\sqrt{2r}}$ for a $2r$ root of unity invariant with $A = e^{\frac{2\pi i}{4r}}$. For these conventions see also \cite{KL}. It is shown in \cite{roberts:95} that $CH(\sfM)$ is independent of the choice of the embedding $E$ (thanks to the insertion of $\Omega$ elements), and the handle decomposition of $M$. It is therefore an invariant of $M$.
}
\end{definition}

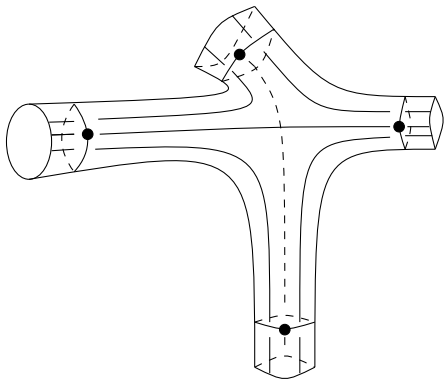
\begin{figure}[H] \label{diag:Omega_insertion}
\begin{center}
\begin{tikzpicture}[every tqft/.append style={transform shape},scale=4]
\draw[rounded corners] (2.3,1.5) ..controls(2.375,1.65).. (2.5,1.7);
\draw[rounded corners,dashed] (2.3,1.5) ..controls(2.4,1.5).. (2.5,1.7);
\draw (1.75,1.25) ellipse [x radius=0.075cm, y radius=0.125cm];
\draw[rounded corners] (2.3,1.5) ..controls(2.5,1.4).. (1.75,1.375);
\draw[rounded corners] (2.5,0.5) ..controls(2.5,1.25).. (1.75,1.125);
\draw[rounded corners] (2.5,0.5) ..controls(2.6,0.45).. (2.7,0.5);
\draw[rounded corners,dashed] (2.5,0.5) ..controls(2.6,0.55).. (2.7,0.5);
\draw[rounded corners] (3.1,1.225) ..controls(2.7,1.225).. (2.7,0.5);
\draw[rounded corners] (3.1,1.225) ..controls(3.135,1.325).. (3.1,1.4);
\draw[rounded corners] (3.1,1.225) ..controls(3.075,1.325).. (3.1,1.4);
\draw[rounded corners] (2.5,1.7) ..controls(2.75,1.4).. (3.1,1.4);
\draw[rounded corners] (3,1.225) ..controls(2.975,1.3).. (3,1.4);
\draw[rounded corners,dashed] (3,1.225) ..controls(3.025,1.3).. (3,1.4);
\filldraw (2.98,1.3) circle (0.5pt);
\draw[rounded corners] (1.9,1.375) ..controls(1.955,1.3) and (1.965,1.22).. (1.9,1.15);
\draw[rounded corners,dashed] (1.9,1.375) ..controls(1.845,1.3) and (1.845,1.22).. (1.9,1.15);
\filldraw (1.945,1.275) circle (0.5pt);
\draw[rounded corners] (2.5,0.65) ..controls(2.6,0.62).. (2.7,0.65);
\draw[rounded corners,dashed] (2.5,0.65) ..controls(2.6,0.68).. (2.7,0.65);
\filldraw (2.6,0.625) circle (0.5pt);
\draw[rounded corners] (2.39,1.45) ..controls(2.45,1.55).. (2.565,1.625);
\draw[rounded corners,dashed] (2.39,1.45) ..controls(2.525,1.49).. (2.565,1.625);
\filldraw (2.45,1.54) circle (0.5pt);
\draw[rounded corners] (2.55,0.65) ..controls(2.55,1.25).. (1.97,1.225);%
\draw[rounded corners] (2.95,1.3) ..controls(2.55,1.3).. (1.985,1.275);
\draw[rounded corners] (2.425,1.485) ..controls(2.55,1.35).. (1.98,1.325);
\draw[rounded corners] (2.53,1.55) ..controls(2.7,1.35).. (2.95,1.35);
\draw[rounded corners] (2.65,0.65) ..controls(2.65,1.25).. (2.95,1.265);%
\draw[rounded corners,dashed] (2.6,0.715) ..controls(2.6,1.35).. (2.5,1.5);
\draw[dashed] (2.6,0.6) -- (2.6,0.55);
\draw (2.55,0.6) -- (2.55,0.48);
\draw (2.65,0.6) -- (2.65,0.48);
\draw (3.,1.35) -- (3.085,1.35);
\draw (3.0075,1.3) -- (3.08,1.3);
\draw (3.,1.27) -- (3.09,1.27);
\draw (1.817,1.315) -- (1.9,1.3175);
\draw (1.825,1.2725) -- (1.9,1.275);
\draw (1.825,1.22) -- (1.9,1.225);
\draw (1.825,1.2725) -- (1.9,1.275);
\draw (2.49,1.595) -- (2.425,1.67);
\draw (2.4,1.505) -- (2.3325,1.5665);
\draw[dashed] (2.5,1.5) -- (2.4285,1.565);
\end{tikzpicture}
\end{center}
\caption{Handlebody component with $\delta$ curves and $\Omega$ elements represented}
\end{figure}

\subsection{The chain mail and TV-invariants}\label{chain-mail-tv}

The key relation between the chain-mail invariant and the TV invariant, which will later use to relate universal quantum computing and TQNNs, is that given a triangulation $T$ of $\sfM$ used to compute $TV(\sfM)$, we can obtain a  handlebody decomposition of $M$ by considering the dual triangulation $T^*$ of $T$. The associated chain-mail invariant arising from $T^*$ is the same as $TV(\sfM)$. The details are found in Roberts \cite{roberts:95}, but we recall some facts that will play a role in the correspondence between topological quantum computation and TQNNs.

Given a triangulation $T$ of $\sfM$, and its dual triangulation $T^*$, we let $D^*$ denote the handlebody decomposition of $\sfM$ where $T^*$ corresponds to $0$ and $1$-handles, and $T$ corresponds to $2$ and $3$-handles. The chain-mail link $C(\sfM,D^*)$ has a configuration where a ball corresponding to a tetrahedron has $3$ attached $1$-handles and $3$ curves through the attaching curves of the handles. Introducing the $\Omega$ elements, for the computation of the corresponding chain-mail invariant (up to normalization constant), we find that the inverse of a $3$-strand fusion rule for $\Omega$ can be applied to each of the $1$-handles. So, each ball will decompose into a tetrahedron for each admissible state. This computation turns out to be the same weight for each admissible state as in the TV-invariant, showing that the latter coincides with the chain-mail invariant.

The relationship between the RT and TV-invariants stated before, can then be obtained through the chain-mail invariant \cite{roberts:95}. The main step toward this direction is the surgery presentation of a $3$-manifold, known as the Lickorish-Wallace Theorem \cite{lickorish3,wallace}. In fact, given a link $L$ whose surgery produces a given $3$-manifold $\sfM$, we can use $\Omega$'s to obtain a skein module element $\Omega L\in \mathcal S S^3 = \mathbb C$, and the value $I(\sfM) = \eta \kappa^{-\sigma(L)\Omega L}$ gives the $\rm{U}_q(SU_2(\mathbb C))$ RT-invariant from skein theory, as constructed by Lickorish. Following an argument on the surgery presentation given in \cite{roberts:95}, and considering the $1$ and $2$-handles attached to the $0$-handles of a given handlebody diagram, one sees that $CH(\sfM) = |I(\sfM)|^2$.

\subsection{Quantum computation and the Turaev-Viro model}\label{qc-tv}

We now show how the TV model simulates modular functors by means of the equivalence with the RT model reviewed in the previous subsections, and how TQNNs compute the single transition amplitudes of such simulations. These results will be later used to show that TQNNs support universal quantum computation.

The modular functor of the Kauffman skein relation is unitary \cite{wenzl1990quantum}, and we can therefore apply the framework of \cite{freedman:02a} to it. We want to use the relation between the RT-model described in Section~\ref{rt-tv} to show that a TQNN can compute transition amplitudes for this class of modular functors.

We consider a diffeomorphism $h : S \longrightarrow S$ of the surface $S$, which can be described by considering its action on the generators of the mapping class group of $S$, $Map(S)$. The diffeomorphism $h$ induces a $3D$ cobordism $H : S \longrightarrow S$ via its mapping cylinder, $H = ([0,1]\times S)\bigsqcup_h S$. The latter, in turn, corresponds to an endomorphism $V(h) : V(S) \longrightarrow V(S)$ via the modular functor $V$, which is associated to it via the RT construction \cite{rt:91}. We can turn the cobordism $H$ into a closed manifold by glueing two regular neighborhoods of two copies of $S$ along $\{0\}\times S$ and   $\{1\}\times S$. We denote by $\sfM_h$ the resulting closed manifold, as this depends on the initial given diffeomorphism $h$ of $S$. Computing the corresponding invariant $I(\sfM_h)$ gives us an element of the ground field $\langle S|\sfM_h| S\rangle$ associated to taking the trace of the operator $V(h): V(S) \longrightarrow V(S)$ associated to the diffeomorphism $h$.

Applying the correspondence described in Section~\ref{rt-tv}, we can compute the invariant $I(\sfM_h)$ by applying the TV construction to a Heegaard diagram of the manifold $\sfM_h$. Such a Heegaard diagram can be obtained from a surgery link presentation of $\sfM_h$, on the link $L$, used to compute $I(\sfM_h)$ by a diagram of $L$ as follows. On a tubular neighborhood of a diagram $D$ of $L$, one adds tubes between two different tubular arcs of each crossing. The corresponding surface is the boundary of the handlebody of the Heegaard splitting. To obtain the diagram, one needs to keep track of the attaching curves of the handles. The $\delta$-curves are obtained as the boundaries of the compressing disks of the tubes added at each crossing, and the surgery longitude of the tubular neighborhood of $L$ and the boundaries of the inner regions of the projected surface are the $\epsilon$-curves. As previously described, from this diagram one obtains a chain-mail link directly by pushing the $\delta$-curves inside the handlebody.
Therefore, one obtains a direct way of computing the chain-mail invariant $CH(\sfM_h) = |I(\sfM_h)|^2$. Up to a proportionality constant from the definitions of Roberts in \cite{roberts:95}, this becomes the TV-invariant.

The $\rm{U}_q(sl(2))$ TQNN \cite{marciano:22, hexagon}, computes the transition amplitudes associated to two boundary spin-networks $N_1$ and $N_2$, associated to some input and output of the quantum algorithm. In other words, each summand of $CH(\sfM_h)$ is computed by a TQNN transition, and summing over all possible configurations for the bulk computes the transition probability. Since the unimodular functor $V$ corresponds to the skein geometric realization of Turaev-Viro \cite{turaevquantum}, we find that $V$ can be simulated by a TQNN upon summing over all possible intermediate transitions.

In the more general setting where the equality of TV and RT models is considered at the functorial level, i.e.~in~\eqref{eqn:TV_RT_functor}, one obtains a similar correspondence where the TV model behaves like the transition probability, while the RT model behaves like the transition amplitude. This fact is at the core of the reason why the TV model is less powerful than the RT model when evaluating closed manifolds. For instance the TV model does not distinguish a manifold $\sfM$ and $\bar{\sfM}$ obtained by reversing the orientation. However, in terms of quantum computing, all that matters is that we can obtain the transition probabilities. Therefore, any concrete implementation of a TQNN that allows to obtain the transition amplitudes directly, gives a valid approach to quantum computing as it allows us to derive the probabilities for the corresponding quantum circuit which is implemented through the mapping cylinder construction described above. We will make use of these facts in subsequent sections.

Lastly, we mention that the correspondence between string-nets and the TV model discovered in \cite{kirillov} suggests an approach to implementing the TV model via string-nets \cite{levin_wen}, which is implementable concretely through condensed matter theory. Therefore, such an implementation would provide a possible concrete universal quantum computer. In fact, by the universality result of \cite{freedman:02a}  one can associate to any target quantum computation a modular functor, and then the transition amplitudes of such functor can be directly computed by performing TQNN transitions on spin-networks as in the description given in this section. Since spin-networks can be supported on string-nets, an implementation of string-nets transitions would provide a concrete computation of the transition amplitudes of the initial quantum process. In other words, this would be a concrete topological quantum computer with universal capabilities.

%%%%%%%%%%%%%%%%%%%%%%%%%%%%%%%%%%%%%%%%

\subsection{The Turaev-Viro code}\label{tv-code}

In this section we consider the formulation of Freedman-Kitaev-Wang \cite{freedman:02a} of UQC. This relies on the RT-invariant, and its reformulation via the TV model, whose theoretical framework was given in the previous subsections. %, will allow us to obtain a direct formulation of topological quantum computing in terms of TQNNs \cite{marciano:22}.
%
%First, we recall the formulation of Freedman-Kitaev-Wang \cite{freedman:02a} of universal quantum computing. This relies on the RT invariant, and its reformulation via the TV model, whose theoretical framework was given in the previous subsection, will allow us to obtain a direct formulation of topological quantum computing in terms of TQNNs \cite{marciano:22}.
Thus, this will enable us to relate TQFTs via the TV-invariant to TQNNs \cite{marciano:22} and UQC, using the correspondences described in the previous two subsections, and proceed accordingly.

The treatment of UQC in \cite{freedman:02a,freedman:02b} relative to the class BQP, can be closely tied to the TV-invariants and TV-code \cite{tv:92}, as these are extensively studied in the context of TQFT \cite{koenig:10}, to which we refer for a suitable review of the theory --- see also \cite{foxon:95}. Let us consider a 3-manifold $\sfM$ with triangulation. It is a striking result that the TV-invariant $TV_{\mathcal{C}}(\sfM)$ (with respect to a spherical category $\mathcal{C}$ \cite{barrett:99}), can be interpreted as a {\em quantum error correcting code} (QECC) \cite{knill:97} as defined on a TQFT.

For the TV-code, one proceeds with taking a compact 2D-surface $\Sigma$, such that $\sfM = \Sigma \times [-1,1]$, then triangulating the boundary $\del \sfM = \Sigma \times \{-1, 1\}$. This latter labeling, denoted $\chi$, consists of a pair $\chi= (\chi_{+}, \chi_{-})$ of labelings of the triangulation of $\Sigma$. The TV-invariant $\mathsf{TV}_{\mathcal{C}}(\sfM, \chi)$ can be regarded as the matrix element of an operator
\begin{equation}\label{tv-code-1}
\mathbb{TV}(\Sigma \times [-1,1]) = \sum_{\chi= (\chi_{+}, \chi_{-})} \mathsf{TV}_{\mathcal{C}}(\sfM, \chi)\vert \chi_{+} \rangle \langle \chi_{-} \vert
\end{equation}
acting on the Hilbert space $(\mathbb{C}^m)^{\otimes \vert E \vert}$ of labelings of the $\vert E \vert$ edges of the triangulation of $\Sigma$
by $m$ anyon labelings of $\mathcal{C}$.

The TV-code on $\Sigma$ then follows as defined to be the range of $\mathbb{TV}(\Sigma \times [-1,1])$ above. This is then specified by the Hilbert space $\mathcal{H}_{\Sigma}$ of equivalence classes of ribbon graphs embedded in $\Sigma$. In this case, the QECC is based on ribbon graphs encoding, for suitable labeling $\ell$, $\mathcal{H}^{\ell}_{\Sigma}$, inside $(\mathbb{C}^2)^{\otimes \vert E \vert}$, where $\oplus_{\ell} \mathcal{H}^{\ell}_{\Sigma} = \mathcal{H}_{\Sigma}$. These facts follow from \cite{tv:92}, and are reviewed in depth in \cite[7.2-7.3]{koenig:10}.

\subsection{Connection with the Reshetikhin-Turaev invariant}\label{rt-inv}

The RT-invariant is an outcome of \cite{rt:91}. This invariant as discussed in \S\ref{rt-tv} {\it et seq}, is in direct correspondence with the TV-invariant
in \cite{roberts:95} (skein theory), and in \cite{kawa:05} (geometry of tube domains). In fact, for the Witten-Reshetikhin-Turaev ($\mathsf{WRT}$) invariant, and the TV-invariant, the two are related by the categorical double $D\mathcal{C}$, namely
\begin{equation}\label{wrt-tv}
\mathsf{TV}_{\mathcal{C}}(\sfM) = \mathsf{WRT}_{D\mathcal{C}}(\sfM)
\end{equation}
See \S 6 and Appendix C in \cite{koenig:10} for details.  From a computational perspective, these invariants are \#P-hard \cite{alagic:17}.

\subsection{TQNNs implement UQC}\label{UQC-main}

In what follows, we will say that a computational scheme \emph{provides quantum processes for UQC} if to each quantum operation it is possible to associate a modular functor (in some suitably defined tensor category), and that the transition amplitudes of the latter can be computed by means of the given computational scheme. With this notation, we can then summarize the results in the previous sections as follows.

\begin{theorem}\label{main-tqnn-1}
A TQNN provides quantum processes for UQC.
\end{theorem}
\begin{proof}
For a primitive fifth root of unity, with $q = e^{\frac{\pi i}{5}}$ (i.e. $A = e^{\frac{\pi i}{10}}$) in the notation of Kauffman and Lins \cite{KL}, the unimodular functor $V$ of RT for the $\rm{U}_q(sl(2))$ invariant is a universal quantum computation \cite{freedmanmodular} (notice that the notation for $q$ and $A$ in \cite{freedmanmodular} is different from \cite{KL}). The functor $V$ can be simulated by computing all possible transitions via a TQNN, since this full sum is the TV-invariant, as shown in Section~\ref{qc-tv}. The details in \S\ref{rt-tv} give an overview of the relation (with relevant references), along with how TQNNs are used to simulate the TV-invariant based on \cite{marciano:22}. As a consequence, TQNNs are quantum processes for UQC.
\end{proof}

Theorem \ref{main-tqnn-1} confirms the intuition that TQFTs, by virtue of implementing QECCs \cite{fgm:24a}, provide the resources needed for UQC. We note, in particular, that a TQNN admitting a nonabelian gauge theory falls into the setting of \cite{freedman:98}, for which such a TQFT supports the counting class \#P. The latter denotes the class of questions, into which NP maps, asking of a given NP algorithm, having fixed polynomial time cut-offs, how many settings of guess hits lead to a ``yes'' response.\footnote{The inclusions $\rm{P} \subseteq \rm{NP} \subseteq{\#P}$ have been conjectured to be strict. As pointed out in \cite{alagic:17}, P and NP are decision classes not formally commensurable with the counting class \#P. Formally, the inclusions are  $\rm{P} \subseteq \rm{NP} \subseteq{P^{\#P}}$.}

%%%%%%%%%%%%%%%%%%%%%%%%%%%%%%%%%%%%%%%%%%%

%%%%%%%%%%%%%%%%%%%%%%%%%%%%%%%%%%%%%%%%%

\section{Amplituhedra for generic computations}\label{scatter-ampl}

The amplituhedron formalism, introduced by Arkani-Hamed and Trnka~\cite{arkani:14} (see also \cite{arkani:16}), represents a paradigm shift for the understanding of scattering amplitudes in planar $\mathcal{N}=4$ supersymmetric Yang-Mills (SYM) theory. This geometric approach transcends the traditional framework of Feynman diagrams and unitarity-based methods by encoding quantum amplitudes as volumes of a novel mathematical object in Grassmannian space. At its core, the amplituhedron generalizes the simple geometric picture of a convex polygon (which underlies the tree-level amplituhedron for $m=2$) to higher-dimensional positive geometries that capture the complete perturbative expansion.

The construction begins \cite{arkani:14} with positive external data encoded in a matrix $Z_I^A \in \text{Mat}_{+}(n,k+m)$ --- i.e. a $n \times (k+m)$ matrix with non-negative entries --- where $A=1,\ldots,k+m$ and $I=1,\ldots,n$, with all ordered maximal minors required to be positive. Here, $n$ denotes the number of external particles, $k$ specifies the helicity sector (N$^{k}$MHV amplitudes), and $m$ is a parameter related to spacetime dimensionality (with $m=4$ corresponding to physical spacetime). We consider now the Grassmannian $\mathrm{Gr}_{k,n}$, the set of all $k$-dimensional linear subspaces of either $\mathbb{R}^n$ or $\mathbb{C}^n$, namely
$\mathrm{Gr}_{k,n} = \{\text{all } k\text{-dimensional linear subspaces of } \mathbb{R}^n {\text{or }}  \mathbb{C}^n\}$.  The amplituhedron $\mathcal{A}_{n,k}^{(m)}$ is a subset of the Grassmannian $\Gr_{k, m+k}$ defined as the image of the positive Grassmannian $\Gr^{>0}_{k +n}$ respecting the relation $Y^A = C_{\alpha}^I Z_I^A$, where $C_{\alpha}^I$ are positive matrices parametrizing points in $G^{> 0}_{k,n}$. This geometric formulation naturally incorporates twistor variables, with the $Z_I^A$ for $m=4$ being interpretable as momentum twistors that combine both spinor helicity variables and spacetime coordinates.

The connection to physical amplitudes emerges through the construction of a canonical differential form $\Omega_{n,k}^{(m)}$, characterized by logarithmic singularities on all the boundaries of the amplituhedron. The scattering amplitude is obtained by integrating this form over an appropriate cycle:
\begin{equation}
\mathcal{A}_{n,k} = \int_{\gamma} \Omega_{n,k}^{(m)}\,.
\end{equation}
Remarkably, this differential form can be constructed through various equivalent approaches: either directly from the geometry via triangulation, through the push-forward of a ``dual'' amplituhedron, or via recursion relations that mirror the Britto-Cachazo-Feng-Witten (BCFW) recursion in twistor space. For tree-level amplitudes, this reproduces known results from the Grassmannian contour integral formulas, while at loop level it generates the complete integrand as a sum over positroid strata of the amplituhedron.

The loop-level amplituhedron extends this picture by introducing additional degrees of freedom that parametrize the loop momenta. The all-loop integrand takes the form:
\begin{equation}
\mathcal{M}_{n,k}^{(L)} = \sum_{\Gamma} \int \prod_{\ell=1}^L \frac{d^4\alpha_\ell}{\text{vol(GL}(1))} \prod_f \frac{1}{f_{\Gamma}(\alpha,Z)}\,,
\end{equation}
where the sum ranges over all valid triangulations of the amplituhedron space, and the functions $f_{\Gamma}$ encode geometric information about each cell in the decomposition. This representation makes several profound properties manifest: locality and unitarity emerge as direct consequences of the geometry rather than being imposed as external constraints, and the intricate cancellation of spurious poles in traditional approaches becomes a natural feature of the positive geometry.

The power of the amplituhedron framework lies in its ability to reveal hidden mathematical structures underlying quantum field theory. The geometry automatically encodes the Yangian symmetry of $\mathcal{N}=4$ SYM, makes the cluster algebra structure of amplitudes transparent, and provides a natural setting for understanding the positive Grassmannian structure observed in scattering amplitudes. Furthermore, the differential form $\Omega_{n,k}^{(m)}$ exhibits remarkable properties under collinear limits and soft limits, with the geometry naturally encoding the expected factorization behavior. This has led to new insights into the infrared structure of gauge theories and has provided fresh perspectives on the universality of certain amplitude structures across different quantum field theories.

From a computational standpoint, the amplituhedron approach offers significant advantages. It eliminates the gauge redundancy inherent in Feynman diagram calculations, avoids the proliferation of terms encountered in traditional methods, and provides a global rather than local description of amplitudes. The geometric formulation also suggests new recursion relations and combinatorial structures that have enabled calculations of amplitudes at higher loop orders and multiplicities that would be prohibitive with conventional methods. Recent work has shown how certain ``negative geometries'' within the amplituhedron framework can naturally capture the cancellation of ultraviolet divergences, hinting at deeper connections between the geometry and renormalization.

The amplituhedron has also inspired new connections between quantum field theory and various areas of pure mathematics. Its positive geometry structure relates to ongoing developments in tropical geometry \cite{Speyer:04} and matroid theory \cite{Ardila:22}, while its recursive properties mirror structures in algebraic combinatorics \cite{Stanley:04}. The differential forms $\Omega_{n,k}^{(m)}$ have been shown to relate to certain classes of hypergeometric functions and polylogarithms that appear in amplitude calculations, suggesting deeper ties to number theory \cite{Neu:99}. Moreover, the geometric perspective has led to new formulations of quantum field theory where spacetime locality emerges from more primitive mathematical structures \cite{Lam:2018mun}, potentially offering new pathways to understanding quantum gravity.

While initially developed for planar $\mathcal{N}=4$ SYM, the amplituhedron concept has been extended to non-planar theories, theories with less extended supersymmetry, and even certain gravitational amplitudes. These generalizations suggest that the amplituhedron may represent a universal feature of quantum field theories \cite{Travaglini:2022uwo}, with different theories corresponding to different geometric realizations. Current research continues to explore the boundaries of this framework, including its connections to celestial holography \cite{Pasterski:2021raf}, its implications for the AdS/CFT correspondence \cite{Arkani-Hamed:2021iya}, and its potential applications to collider physics phenomenology \cite{Bern:2022jnl}.

\subsection{Amplituhedron and Topological Quantum Field Theory}
\label{sec:amplituhedron_tqft}

The amplituhedron formalism reveals deep connections with TQFTs, particularly through their shared mathematical structures in cohomology and geometric quantization. These connections emerge most clearly when examining the amplituhedron's differential form $\Omega_{n,k}^{(m)}$, which bears strong resemblance to partition functions in topological field theories. The fundamental link arises from viewing the amplituhedron as a positive geometry whose canonical form defines a cohomology class on the space of kinematical data --- a perspective that mirrors how TQFTs produce topological invariants via integration of characteristic forms over moduli spaces \cite{Witten:1988ze}. Specifically, the amplituhedron's volume form can be interpreted as a generalized period integral of the form
\begin{equation}
\int_\Gamma e^{F(\mathcal{Z})} \, \Omega(\mathcal{Z})\,,
\end{equation}
where $F(\mathcal{Z})$ is a function encoding the positive geometry constraints and $\Omega(\mathcal{Z})$ is a holomorphic form, directly analogous to the correlation functions in holomorphic topological field theories. This connection becomes precise when considering that both frameworks employ similar mathematical tools: the positive Grassmannian $G^{> 0}_{k,n}$ underlying the amplituhedron appears as the moduli space of solutions to certain equations in topological string theory, while the stratification of the amplituhedron into positroid cells parallels the cell decomposition of moduli spaces in TQFT.

The relationship extends further when examining the quantum equations of motion. In topological string theory, the physical states are represented by cohomology classes satisfying $Q_B|\Psi\rangle = 0$, where $Q_B$ is the BRST operator \cite{Witten:1988ze}. Within the amplituhedron framework, the scattering equations
\begin{equation}
\mathcal{E}_a = \sum_{i=1}^n \frac{s_{ai}}{\sigma_i} = 0\,,
\end{equation}
with $\sigma_i$ being the puncture locations on a Riemann surface, similarly define a cohomological problem whose solutions determine the support of the amplitude. This suggests an underlying correspondence where the amplituhedron's geometry encodes a ``hidden'' topological symmetry analogous to BRST symmetry in TQFT \cite{Birmingham:1991ty}. The connection becomes particularly striking in the context of twistor string theory, where the amplituhedron for $m=4$ emerges from the topological B-model on super-twistor space --- a topological string theory whose open string sector describes $\mathcal{N}=4$ SYM. In this setting, the amplituhedron volume form directly computes certain correlation functions in the topological string, with the positive geometry constraints implementing the selection rules for non-vanishing amplitudes.

Recent developments have shown that the amplituhedron's differential forms $\Omega_{n,k}^{(m)}$ naturally organize into complexes that mirror the BRST complexes of topological field theories. The logarithmic singularities of $\Omega$ along amplituhedron boundaries correspond to the descent equations in TQFT, with the residue theorems governing amplitude factorization being directly analogous to the Ward identities in topological theories. Moreover, the combinatorial structure of amplituhedron triangulations exhibits formal similarities with state sum invariants in TQFT --- where the latter assign algebraic data to simplices in a triangulation of spacetime, the former assigns geometric data to cells in the positive Grassmannian. This suggests that the amplituhedron may be understood as a novel kind of ``positive'' TQFT where traditional topological invariance is replaced by a more subtle positivity-preserving structure.

The connection extends to computational techniques as well. The Grassmannian integral representation of scattering amplitudes, namely
\begin{equation}
\mathcal{A}_{n,k} = \int_{\gamma\subset G(k,n)} \frac{d^{k\times n}C_{\alpha i}}{\text{vol(GL}(k))} \prod_{\alpha=1}^k \delta^4|4(C_{\alpha i}\mathcal{Z}_i^A) \,,
\end{equation}
bears strong resemblance to correlation function computations in topological matrix models, particularly in their shared use of localization techniques. Both frameworks benefit dramatically from the reduction of infinite-dimensional path integrals to finite-dimensional geometric problems --- in TQFT via topological twisting, and in the amplituhedron via positive geometry. This parallel has led to new insights into the geometric nature of quantum field theory, where traditional dynamical principles are replaced by geometric constraints.

Looking beyond $\mathcal{N}=4$ SYM, the amplituhedron-TQFT correspondence suggests a unified perspective on quantum field theories with different amounts of supersymmetry. Certain deformations of the amplituhedron geometry have been shown to correspond to topological twists of less supersymmetric theories, while the purely geometric formulation avoids any reference to spacetime locality --- a feature shared with topological field theories. This has inspired new approaches to constructing quantum field theories where traditional dynamical principles emerge from purely topological starting points, with the amplituhedron serving as a prime example of how ``spacetime-less'' structures can nevertheless encode physical scattering processes \cite{arkani:14}. The ongoing exploration of these connections continues to yield new mathematical structures at the interface of algebraic geometry, representation theory, and quantum field theory, suggesting that the amplituhedron may represent a special case of a broader class of ``geometric quantum field theories'' that include both traditional TQFTs and conventional physical theories as particular limits.

\subsection{The Grassmannian and hypersimplex}\label{grass-1}

We will elucidate the Grassmannian model in relationship to the amplituhedron and hypersimplex.
Let us proceed with $G_k(\mathbb{C}^n)$, which we recall is an analytic manifold \cite{wong:67} (considerations for $G_k(\mathbb{R}^n)$ are similar).
It is well known that commencing from a torus $\mathbb{T}^n = (\mathbb{C}^*)^n$-action on $G_k(\mathbb{C}^n)$, the image of the moment map
\begin{equation}\label{mpment-1}
\mu : G_k(\mathbb{C}^n) \lra \mathbb{R}^n
\end{equation}
is the {\em hypersimplex} $\Delta_{k,n} \subset \mathbb{R}^n$, which is a $(n-1)$-dimensional convex polytope \cite{gelfand:87}. Hypersimplices $\Delta_{k,n}$ are examples of alcoved polytopes \cite{lam:07}.  In \cite[\S 3.6]{vaid:22} it is shown, via spin networks, that the kinematic space of states for the scattering of $n$-massive particles, is identifiable with the space of semi-classical states of quantum gravity determined by Freidel-Levine (FL) intertwiners, where the space of unique labels of the FL-coherent states is precisely the Grassmannian $G_2(\mathbb{C}^n)$.  In fact, \cite[\S 4.4]{vaid:22} establish a one-to-one correspondence between the scattering amplitudes of $n$-massless particles, and a single $n$-vaued spin-network. Thus, commencing from a spin-network, and applying the moment map technique to this Grassmaniann, leads to a $(n-1)$-dimensional polytope which is the hypersimplex.\footnote{Spin-networks can be associated with hypersimplices in a number of ways. For instance, a spin-network representing intertwiners can be associated with a 2D-Ising model \cite{dittrich:16}, which is in turn associated with a hypersimplex \cite{guevara:20}.} $\Delta_{2,n}$.

Here, we are particularly interested in $G^{\geq 0}_{k,n}$, the non-negative Grassmannian, which is a subset of $G_k(\mathbb{R}^n)$ in which all  Pl\"{u}cker coordinates are non-negative. The Grassmannian $ G^{\geq 0}_{k,n}$ is a prototypical model for the theory of scattering amplitudes, and in particular, for the amplituhedron \cite{arkani:14,arkani:16} (for a general overview of the theory in relationship to hypersimplices, see e.g. \cite{lukowski:23,parisi:23}).  It was pointed out in \cite{parisi:23} that restricting the moment map $\mu$ accordingly, the image remains the hypersimplex $\Delta_{k,n}$ \cite{tsuker:15}.

In underscoring the role of a spin-network $\Phi$ in the theory of scattering amplitudes (and amplituhedra), we may exemplify the construction by taking the appropriate
spin-network $\Phi$ of a TQNN to be an embedded FDG in $G^{\geq 0}_{k,n}$ as in \S\ref{spin-comp-1}.
We can then proceed to consider:

\begin{itemize}
\item[(1)]  {\bf The Amplituhedron Map} \cite{arkani:14,arkani:16}: Any $n \times (k+m)$ matrix $Z$ with maximal minors positive, induces a map $Z: G^{\geq 0}_{k,n} \lra G_k(\mathbb{R}^{k+m})$, whose image has full dimension $k(n-k)$, and defines the amplituhedron $\mathcal{A}_{n,k,m}(Z)$ \cite{lukowski:23,parisi:23}, studied for the case $m=2$.

    \item[(2)] {\bf $T$-Duality correspondence}: Many of the important combinatorial methods used in scattering theory can be traced back to Postnikov's original work (see \cite{lam:07,lam:24,postnikov:18}) that includes the key feature of {\em positroid tilings} (see also \cite{ardila:14}). Much has been written about these tilings (cf. the BCFW triangulations in \cite{even:21}). The techniques in question lead to a {\em $T$-Duality Correspondence} between such tilings of hypersimplices and those of amplituhedra. For instance, the hypersimplex $\Delta_{k+1,n}$ and the amplituhedron $\mathcal{A}_{n,k,2}(Z)$ are combinatorially dual to each other \cite{lukowski:23,parisi:23}.
\end{itemize}

Particularly relevant for our development of ideas is the {\em Momentum Amplituhedron} $\mathcal{M}_{n,k}$ \cite{damg:19}, for which there is a kinematically determined map, namely
\begin{equation}\label{moment-ampl}
\Psi: G^{> 0}_{k,n} \lra \mathcal{M}_{n,k}\,.
\end{equation}
The amplituhedron $\mathcal{M}_{n,k}$ is a positive geometry whose canonical logarithmic differential encodes scattering amplitudes in spinor helicity variables, more specifically, computes tree level scattering amplitudes via the tree-level amplituhedron $\mathcal{A}^{\rm{tree}}_{n,k}$ in N=4 SYM momentun space.

In relationship to scattering amplitudes, let us summarize the construction: we have shown here that a spin-network $\Phi$ of a TQNN, where the former is formally represented by a FDG, when embedded into a Grassmannian, can be associated with a hypersimplex, and hence to an amplituhedron.  The combinatorial type of the various Grassmannians (more generally, flag varieties) and those of amplituhedra, previously conjectured by Postnikov and others, has been determined in terms of triangulations, simplicial methods, and the topology of CW complexes (see e.g. \cite{hatcher:02} for the relevant techniques).
\medn
{\bf Examples:}
\begin{itemize}
\item[(1)] $G^{\geq 0}_{k,n}$ is homeomorphic to a $k(n-k)$-dimensional closed ball. More generally, for a given partial flag variety $G/P$, where $G$ a complex semi-simple Lie group, and $P$ a parabolic subgroup, there is a corresponding $(G/P)^{\geq 0}$, whose cell decomposition is a (regular) closure- finite weak-topology (CW-) complex. Thus, the closure of each cell is homeomorphic to a closed ball of some dimension \cite{galashin:22a,galashin:22b}.

\item[(2)] Recursively defined Britto-Cachazo-Feng-Witten (BCFW) cells in $G^{\geq 0}_{k,n}$ triangulate the amplituhedron $\mathcal{A}_{n,k,4}$.

\item[(3)] For $k+m= n$, we have $G^{\geq 0}_{k,n} \cong \mathcal{A}_{n,k,m}$ \cite[\S 0.4]{bao:19}.

\end{itemize}
In general, a TQFT can be modeled on a (finite) dimensional Hilbert manifold, as shown in \cite[\S 4]{fgm:22a}. Such a manifold can be considered as a CW-complex, and can also be triangulated (again, see \cite{hatcher:02}). For instance, the cellular homeomorphism type of $G^{\geq 0}_{k,n}$ in (1) above, can be recovered when the CW-complex type of the TQFT is specified in terms of closed cells of dimension $k(n-k)$.

\subsection{Amplituhedron and Spinfoam Formalism: Geometric Connections}
\label{sec:amplituhedron_spinfoam}

The amplituhedron framework exhibits profound connections to spinfoam models of quantum gravity, revealing deep geometric parallels between scattering amplitudes in gauge theory \cite{arkani:14,Elvang:2015rqa,Vaid:2022lrm} and discrete approaches to quantum spacetime \cite{Regge_calculus,Freidel:2005qe}. At the conceptual level, both frameworks replace traditional dynamical variables with combinatorial structures encoding the geometry of interactions: where spinfoams describe quantum spacetime as a network of discrete volumes representing quantized geometry, the amplituhedron encodes particle interactions through positive geometries in Grassmannian space.

This correspondence becomes particularly evident when examining the mathematical structures underlying both approaches. For instance, the spinfoam partition function for loop quantum gravity \cite{Rovelli_book, Rovelli:2014ssa} takes the form
\begin{equation}
Z = \sum_{j_f,I_e} \prod_f A_f(j_f) \prod_v A_v(j_f,I_e) \,,
\end{equation}
where the sum runs over spins $j_f$ and intertwiners $I_e$ labeling the faces and edges of the spinfoam complex, while the amplituhedron's volume form $\Omega_{n,k}^{(m)}$ similarly involves a sum over positroid cells with each cell contributing a geometric factor. Both formulations emphasize the fundamental role of positivity constraints: just as the amplituhedron requires positive Grassmannian configurations, modern spinfoam models implement simplicity constraints that ensure geometricity through positivity conditions on the representation labels.

The technical implementation reveals further striking similarities.  For instance, for lower-dimensional spin-foam models --- either the Ponzano-Regge model originated from the path-integral quantization of the Euclidean three-dimensional Einstein-Hilbert action or the Turaev-Viro model, connected to the path-integral quantization of the Euclidean three-dimensional Einstein-Hilbert action plus cosmological constant \cite{Gresnigt:2022lwq} --- the topological invariance and the Biedenharn-Elliott identity imply recursion relations for the 6j-symbols \cite{Schulten:1975yu} that mirror the BCFW \cite{Britto:2005fq} recursion relations in amplituhedron calculations, with both employing complex-analytic techniques to construct physical quantities from simpler building blocks. Recent work has shown that certain limits of spin-foam amplitudes reproduce the twistor space structures central to the amplituhedron, particularly when considering the null boundary geometry of spinfoam vertices \cite{Bianchi:2021ric,Skinner:2010cz}. This connection becomes concrete in the study of scattering amplitudes on quantum flat space, where the semiclassical limit of spin-foam models on a boundary graph reproduces momentum space amplitudes that can be matched to amplituhedron predictions.

The geometric interpretation provides a unifying perspective: where spin-foam vertices represent quantized 4-simplices encoding spacetime curvature, amplituhedron cells correspond to interaction vertices encoding particle scattering. Both approaches implement a form of holographic principle --- the spin-foam boundary states live on a 3D surface while encoding 4D geometry, paralleling how the amplituhedron projects 4D scattering data onto lower-dimensional positive geometries. The combinatorial nature of the amplituhedron's positroid stratification finds its counterpart in the spin-network combinatorics of spin-foam models, with both employing sophisticated diagram.

\subsection{Amplituhedron and BF Theories in Arbitrary Dimensions}
\label{sec:amplituhedron_bf_generic}

The profound connection between amplituhedron geometry and topological BF theories extends to arbitrary spacetime dimensions, revealing a unified framework where positive geometry encodes gauge-theoretic structures. For a $D$-dimensional BF theory with gauge group $G$ on a manifold $\mathcal{M}_D$, the action \cite{Birmingham:1991ty}

\begin{equation}
S_{\text{BF}}^{(D)} = \int_{\mathcal{M}_D} \Tr\left(B \wedge F + \frac{\kappa}{D-2} B \wedge B \wedge \underset{D-2}{\cdots} \wedge B\right)\,,
\end{equation}

where $B$ is a $\mathfrak{g}$-valued $(D-2)$-form and $F=dA + A \wedge A$, finds its amplituhedron counterpart in the positive geometry $\mathcal{A}_{n,k}^{(D+1)} \subset G^{>0}_{k,k+D+1}$ defined through momentum twistors $\mathcal{Z}_I^A \in \text{Mat}_+(n,k+D+1)$. The dimensional correspondence arises because the amplituhedron for $m=D+1$ encodes the kinematic space of $D$-dimensional scattering amplitudes \cite{Arkani-Hamed:2017vfh}.

The quantum equivalence manifests when comparing the BF theory path integral

\begin{equation}
Z_{\text{BF}}^{(D)} = \int \mathcal{D}A \mathcal{D}B \, e^{iS_{\text{BF}}^{(D)}}\,,
\end{equation}

with the amplituhedron's volume integral \cite{Damgaard:2019ztj}

\begin{equation}
\mathcal{A}_{n,k}^{(D+1)} = \int_{\Gamma \subset G_+(k,n)} \frac{d^{k \times n} C_{\alpha i}}{\text{vol(GL}(k))} \prod_{\alpha=1}^k \delta^{(D+1)|(D+1)}\left(C_{\alpha i} \mathcal{Z}_i^A\right) \prod_{j=1}^{n} \frac{1}{\langle C_{1\cdots k} \rangle^{\gamma_j}}\,,
\end{equation}

where $\gamma_j$ are exponents determined by the dimension and $\langle C_{1\cdots k} \rangle$ denotes the $k \times k$ minor. Both integrals localize onto special geometries: flat connections for BF theory and positive Grassmannian subspaces for the amplituhedron.

The constraint structure exhibits striking parallels. For instance, the BF theory's equations of motion, namely

\begin{align}
F &= 0 \,,\\
d_A B &= 0\,,
\end{align}

find their geometric counterpart in the amplituhedron's boundary equations \cite{Ferro:2020lgp}, i.e.

\begin{equation}
\langle Y i_1 \cdots i_{D+1} \rangle = 0\,, \quad \text{for certain cyclic sets } \{i_j\}\,,
\end{equation}

which define the co-dimension one boundaries of $\mathcal{A}_{n,k}^{(D+1)}$. These correspond to factorization channels where intermediate particles go on-shell.

The connection becomes algebraic when examining the deformation quantization. For BF theory, the Poisson structure

\begin{equation}
\{B_{\mu\nu}^a(x), B_{\rho\sigma}^b(y)\} = \epsilon_{\mu\nu\rho\sigma} \delta^{ab} \delta^{(D)}(x-y)
\end{equation}

quantizes to a non-commutative algebra of observables, mirroring the non-commutative cluster algebra structure of amplituhedron coordinates \cite{Arkani-Hamed:2017mur}, namely

\begin{equation}
x_i x_j = q^{2\epsilon_{ij}} x_j x_i\,,
\end{equation}

where $\epsilon_{ij}$ is the skew-symmetric matrix defining the exchange relations in the positive Grassmannian.

In even dimensions $D=2m$, the correspondence specializes beautifully. The amplituhedron $\mathcal{A}_{n,k}^{(2m+1)}$ encodes amplitudes whose factorization properties match the descent equations of $2m$-dimensional BF theory \cite{Cattaneo:2023wjq}, i.e.

\begin{equation}
d_{\text{CE}} \omega_{2m+1} = \kappa \omega_3 \wedge \omega_{2m-1}\,,
\end{equation}

where $d_{\text{CE}}$ is the Chevalley-Eilenberg differential and $\omega_p$ are $p$-forms in the BRST complex. This is mirrored by the amplituhedron's differential form satisfying

\begin{equation}
d\Omega_{n,k}^{(2m+1)} = \sum_{\text{residues}} \Omega_{L} \wedge \Omega_{R}\,.
\end{equation}

The most profound connection emerges in the context of holography. For $D$-dimensional BF theory on asymptotically AdS spacetimes, the boundary dual is a $(D-1)$-dimensional topological theory whose phase space matches the kinematic space of the $m=D$ amplituhedron \cite{Costello:2023hmi}. This suggests that

\begin{equation}
Z_{\text{BF}}^{(D)}[\partial\mathcal{M}] \underset{\text{AdS}}{\longleftrightarrow} \mathcal{A}_{n,k}^{(D)} \qquad \text{as } n \to \infty\,,
\end{equation}

where the left side denotes the boundary partition function and the right side the large-$n$ limit of amplituhedron volumes.

\subsection{Amplituhedron and TQFT in 3D: the Turaev-Viro Model}
\label{sec:amplituhedron_bf}

The connection between the amplituhedron and three-dimensional topological field theories becomes particularly transparent when examining the case of BF theory and its quantization via the Turaev-Viro state sum model. In three spacetime dimensions, pure BF theory with gauge group $SU(2)$ has the action \cite{Birmingham:1991ty}

\begin{equation}
S_{\text{BF}}^{(3)} = \int_{\mathcal{M}_3} \Tr\left(B \wedge F\right)\,,
\end{equation}

where $B$ is an $\mathfrak{su}(2)$-valued 1-form and $F = dA + A \wedge A$ is the curvature of the connection $A$. Remarkably, the corresponding $m=4$ amplituhedron $\mathcal{A}_{n,k}^{(4)}$ encodes scattering amplitudes whose kinematic space matches the moduli space of flat connections that solve the BF equations of motion $F=0$, $d_A B=0$ \cite{Arkani-Hamed:2017vfh}.

The quantum mechanical link emerges when comparing the Turaev-Viro partition function \cite{Turaev:1992hq}, namely

\begin{equation}
Z_{\text{TV}} = \sum_{j_f \in \text{Irrep}(\rm{U}_q(\mathfrak{sl}(2)))} \prod_{\text{edges}} \dim_q(j_f) \prod_{\text{tetrahedra}} \left\vert
\begin{array}{ccc}
j_1 & j_2 & j_3 \\
j_4 & j_5 & j_6 \\
\end{array}
\right\vert_q\,,
\end{equation}

with the amplituhedron's volume form for $m=4$, which is

\begin{equation}
\Omega_{n,k}^{(4)} = \int_{\gamma \subset G_+(k,n)} \frac{d^{k \times n} C_{\alpha i}}{\text{vol(GL}(k))} \prod_{\alpha=1}^k \delta^{4|4}\left(C_{\alpha i} \mathcal{Z}_i^A\right) \prod_{j=1}^n \frac{1}{\langle C_{1\cdots k} j \rangle}\,.
\end{equation}

The critical observation is that both constructions rely on a decomposition into elementary geometric building blocks: tetrahedra for Turaev-Viro, and positroid cells (positroid polytopes) for the amplituhedron --- with quantum amplitudes being assembled from combinatorial data attached to these fundamental units.
The constraints in both theories exhibit deep parallels. In the Turaev-Viro model, the admissibility conditions for edge labels $(j_f)$ correspond to the quantum Clebsch-Gordan conditions

\begin{equation}
\vert j_1 - j_2 \vert \leq j_3 \leq \min(j_1 + j_2, k/2 - j_1 - j_2)\,,
\end{equation}

where $k$ is the level related to the deformation parameter $q = e^{2\pi i/(k+2)}$. In the $m=4$ amplituhedron, the analogous conditions are the positivity constraints on ordered minors $\langle Y i i+1 j j+1 \rangle > 0$ for all $i,j$. Both systems implement a form of discreteness: the Turaev-Viro model through quantization of geometric quantities, and the amplituhedron through the stratification of $G^{>0}_{k,n}$ into positroid cells \cite{Postnikov:2006kva}.

The dynamical content emerges through similar mechanisms. In BF theory, the quantum constraints implement flatness of the connection, while in the amplituhedron the global residue theorems enforce the cancellation of spurious poles. These are implemented through the twistor-space localization formula, namely

\begin{equation}
\Omega_{n,k}^{(4)} = \int \prod_{\ell=1}^L \frac{d^2\sigma_\ell}{\text{vol(GL}(2))} \prod_{i=1}^n \delta^{4|4}\left( \mathcal{Z}_i - \sum_{\ell} C_{i\ell}(\sigma) \mathcal{W}_\ell \right)\,,
\end{equation}

where $\mathcal{W}_\ell$ are reference twistors and $C_{i\ell}(\sigma)$ parametrize the Grassmannian (cf. \cite{hodges:15}). This mirrors the integration over the moduli space of flat connections in the Turaev-Viro model when expressed in terms of holonomies.

At the quantum algebraic level, the connection manifests through the appearance of quantum group structures in both frameworks. The Turaev-Viro model's reliance on the representations of $\rm{U}_q(sl(2))$ finds its counterpart in the quantum cluster algebra structure of the amplituhedron's coordinate ring \cite{Fock:2003xx}. Specifically, the canonical form $\Omega_{n,k}^{(3)}$ can be expressed as

\begin{equation}\label{canonical-form}
\Omega_{n,k}^{(3)} = \prod_{i=1}^{3n-6} \frac{dx_i}{x_i} \wedge \prod_{j=1}^k \delta\left( y_j - \prod_{i=1}^{3n-6} x_i^{a_{ij}} \right)\,,
\end{equation}

where the exponents $a_{ij}$ define a quantum torus algebra with relations $x_i x_j = q^{2 \epsilon_{ij}} x_j x_i$. Crucially, this leads to identifying the deformation parameter $q$ of \S\ref{q-circuits} in the Turaev-Viro model when $\epsilon_{ij}$ is the skew-symmetric form defining the cluster algebra (cf. Theorem \ref{main-tqnn-2} below).

Recent developments have shown that under certain limits (taking $q \to 1$ in Turaev-Viro and removing the positivity constraints in the amplituhedron), both theories reduce to their classical counterparts: 3D gravity with vanishing cosmological constant for the Turaev-Viro model, and the classical Grassmannian $G_{k,n}(\mathbb{C})$ for the amplituhedron. This suggests that the amplituhedron may represent a ``positive'' deformation of traditional topological field theories, where the usual topological invariance is replaced by the preservation of positive geometry structures.

\subsection{TQNN execution traces are described by amplituhedra \label{exec-th}}

In the situation of Diagram~\eqref{comp-def}, we consider now the introduction of a scattering process. In this case, we let $\mathscr{N}$ denote the von Neumann algebra corresponding to the process, and we let $S(\mathscr{N})$ denote its set of states. Using the same notation as in Section~\ref{exec-trace}, we let $S(M)$ denote the states of the given system represented by the algebra $\mathscr{N}$, and we let $Prob(x) \cong S(l^\infty(X))$ be the probabilities on the finite data type space. We let $\mathscr P^M(\Delta t)$ and $\mathscr P^{\mathscr{N}}(\Delta t)$ denote the propagators on $S(M)$ and $S(\mathscr{N})$, respectively. For the scattering process to represent the system if at each time we have commutative triangles
    \begin{equation}\label{diag:comm_triangle_scattering}
    \begin{tikzcd}
    S(M)\arrow[rr,"\xi"]\arrow[dr,"\psi"] & & Prob(X)\arrow[dl,"\phi"]\\
    & S(\mathscr{N}) &
    \end{tikzcd},
    \end{equation}
    where $\xi$ is as in Diagram~\eqref{comp-def}, and commutativity is required in the category of sets. At two different time points $t_i$ and $t_f$ with $\Delta t = t_f - t_i$ we require the commutativity of the diagram
    \begin{equation}
    \begin{tikzcd}
    & Prob(X)\arrow[rr,"f"]\arrow[dd] & & Prob(X)\arrow[dd]\\
    S(M)\arrow[rr,"\mathscr P^M(\Delta f)",crossing over]\arrow[ur]\arrow[dr]& & S(M)\arrow[ur]\arrow[dr] & \\
    & S(\mathscr{N}\arrow[rr,"\mathscr P^{\mathscr{N}}(\Delta f)"]) & & S(\mathscr{N})
    \end{tikzcd}.
    \end{equation}

Neglecting gravity, and setting aside anomalies to assume that the SM (standard model) provides a complete description of realizable physical interactions, the structural representation of TQNNs in terms of amplituhedra discussed above can be shown to be well-defined as follows:

\begin{theorem}\label{main-tqnn-2}
The execution trace of any quantum computation, on any input for which that computation is defined, as computed by a TQNN, is described by a unique amplituhedron, and hence corresponds, in the pure initial-state limit, to the unique scattering process for which amplitudes are given by that amplituhedron, and conversely.
\end{theorem}

\begin{proof}
We have seen that both scattering -- in the case of the KLM model (see \S \ref{optics}), even single-photon scattering -- and TQNNs both execute UQC.  Hence for any computable function $f$, and for any input $x$ on which $f$ bis defined, there exist both a scattering process and a TQNN that compute $f(x)$.  In the limit in which $x$ is encoded by a pure state $\ket{i}$, the execution trace of the computation of $f(x)$ is a path in the relevant Hilbert space implemented by a unique (up to isomorphism) unitary $U_{f(x)}$, and the amplitudes of the final state $\ket{f}$ encoding $f(x)$ are all well-defined.  This $U_{f(x)}$ is, therefore, described by a unique (up to isomorphism) amplituhedron.

%It is sufficient to show that the representation of UQC as implemented by TQNNs can be mapped to the generic structure of amplituhedra and vice versa; the relevant maps for each particular initial and state of each particular TQNN then follow as instances.  There are two pertinent formal correspondences \textcolor{blue}{here to consider. Firstly, from the discussion above, a correspondence can be realized between the Turaev-Viro model's reliance on the representations of $\rm{U}_q(sl(2))$, and the quantum cluster algebra structure of the amplituhedron's coordinate ring, on identifying the exponents of the latter --- see \eqref{canonical-form} --- with the deformation parameters of the former.}

%\textcolor{blue}{Secondly, from \S \ref{gp-scatter} we have on the one hand GP as inducing the essential gating for QECC, and hence for UQC as provided by a TQNN, and on the other, GP as ubiquitous in scattering processes. A further correspondence can be realized via transformations between the respective GPs in terms of angular measurement, as can implemented by standard geometric transformations such as rotation, translation, dilation, conformal transformation, and compositions of these:
%$$
%\rm{GP~for~(quantum)~gating} ~ \overset{\rm{geometric~ transformation}}\Longleftrightarrow~ \rm{GP~for~scattering}
%$$
%For the general case of topological phases a transformation between these can in principle be achieved by means of a suitable homotopy of the paths along which the respective phases are defined  (see e.g. \cite{aguilar:20}).}
\end{proof}
%\textcolor{blue}{We have used the term `corresponding to' in the general sense of the term, while not claiming this correspondence to be bijective (i.e. one-to-one and onto); in practice it could conceivably be a several-to-one correspondence.}

The sections above show how to construct the required amplituhedron given an execution trace of a TQNN.  Using a re-labeling of Diagram \eqref{comp-scat} as a base, we can summarize this as follows:
\begin{equation} \label{comp-scat-labeled}
\begin{gathered}
\begin{tikzpicture}
\node at (0,.4) {$TQNN$};
\draw [thick, ->] (-2,0) -- (2,0);
\node at (-2.2,0) {$\ket{i_C}$};
\node at (2.4,0) {$\ket{f_C}$};
\node at (-0.4,-2) {$\rho_M(t_i)$};
\node at (4,-2) {$\rho_M(t_f)$};
\draw [thick, ->] (0.3,-2) -- (3.3,-2);
\node at (1.8,-2.4) {$\mathscr{P}_M(t)$};
\draw [thick, ->] (-0.4,-1.7) -- (-2,-0.3);
\node at (-1.7,-1) {$Q_C$};
\draw [thick, ->] (3.8,-1.7) -- (2.3,-0.3);
\node at (2.4,-1) {$Q_C$};
\draw [thick] (-0.3,-1.7) -- (-0.7,-0.1);
\draw [thick, ->] (-0.8,0.1) -- (-1,1);
\node at (-0.2,-0.5) {$Q_S$};
\node at (-1,1.3) {$\ket{i_S}$};
\draw [thick, ->] (3.9,-1.7) -- (3.1,1);
\node at (4,-0.5) {$Q_S$};
\node at (3.2,1.3) {$\ket{f_S}$};
\draw [thick, ->] (-0.4,1.3) -- (2.5,1.3);
\node at (1,1.6) {$Amp$};
\draw [thick, dashed, ->] (-2.1,0.2) -- (-1.2,1);
\node at (-2.2,0.7) {$Imp$};
\draw [thick, dashed, ->] (2.2,0.3) -- (2.9,1);
\node at (2,0.7) {$Imp$};
\node at (0.2,3) {$TV$ hypersimplex};
\node at (-1.7,2.1) {$TV$ triangulation};
\draw [thick, ->] (0, 0.7) -- (0.1, 2.8);
\draw [thick, ->] (0.3, 2.8) -- (0.9,2);
\node at (1.7, 2.4) {$T$-duality};
\end{tikzpicture}
\end{gathered}
\end{equation}
Subject to the relevant assumptions, Thm. \ref{main-tqnn-2} shows that the map $TQNN \longrightarrow Amp$ is a bijection, in which case the induced $Imp$ maps for the initial and final states are also bijections.  As the proof of Thm. \ref{main-tqnn-2} depends only on the notion of execution trace being well-defined, it is architecture-independent, so applies to any physical model of UQC.\footnote{Note that Theorem \ref{main-tqnn-2}
along with diagram \eqref{comp-scat-labeled} above, provide a representation of TV-tetrahedra, viewed as simple examples of convex polytopes, on the hypersimplices associated with various amplituhedra. In principle, corresponding mappings can be realized within the framework of maps of convex polytopes (see e.g. \cite{grunbaum:03}).
Interestingly, we note the amplituhedron has been extended to the combinatorial topology of the {\em associahedron} \cite{abhy:18} (see also \cite{loday:04,pilaud:23}), from which there is a further extension to the concept of {\em cosmohedron} \cite{afv:25}.}

%%$$
%%{\rm{TV~tetraheda}} \rightleftarrows
%%{\rm{postitroid~(cells)~polytopes}}
%%$$
%%as defined in terms of maps of cell-complexes %%\cite{hatcher:02}.}}

\section{Conclusion} \label{concl}

The representation of TQNNs in terms of  amplituhedra exhibited above, underlies the operational indistinguishability of computation and scattering discussed in \S \ref{locc}, particularly with respect to Diagram \eqref{comp-scatter}. Overall, our study shows that QIP and high energy physics are closer than was previously thought, by showing how an amplituhedron can effectively represent the amplitudes of any generic quantum process, not just scattering (cf. Theorem \ref{main-tqnn-2}). A salient feature here is that the operational setting enforced by LOCC puts ``physics in a box'' by forcing all models to be device independent \cite{grinbaum:17}.  This operational constraint prevents, for example, the roles of entanglement or other nonlocal resources in quantum computation from direct observation by ``classical'' observers, i.e. observers restricted to LOCC \cite{fgmz:25}.  In certain cases, a theoretical QRF-invariant can be defined \cite{cepollaro:25}, but such a quantity cannot be classically/empirically realized.

This representation also points the TQNN, with its spin-networks, towards associated kinematic/momentum amplituhedra, thus realizing scattering processes. The kinematic property can be further linked to the {\em momentum/complexity correspondence} of \cite{barbon:20}. This is significant when quantum computations are viewed in terms of computational complexity, which can thus be assigned an element of momentum.

We note, however, that scattering experiments are performed and analyzed under an assumption of sufficient isolation of ``system'' from ``environment'' to render initial and final states well-defined and the scattering process itself unitary, and that accelerators are designed to assure that this assumption is justified.  The extent to which these assumptions can be justified for generic quantum processes, and particularly for practical quantum computers, remains an open question.  Amplituhedra are well-defined only if the amplitudes they represent are well-defined.  Whether {\em unique} amplituhedra can be assigned to a process depends of whether unique amplitudes can be assigned, which in turn depends on state purity and hence the level of effective coarse-graining.  As Hawking \cite{hawking-incompleteness} and others have pointed out, obtaining both self-consistent and complete descriptions of many systems may not be possible (cf. \cite{abel:25} and references therein, as part of an ongoing debate that questions whether large colliders have the capability
to measure entanglement and detect violation of the Bell inequality).

We have previously shown that TQNNs can implement deep learning (DL), and indeed that classical DL systems generalize because fully-connected multi-layer networks are classical limits of TQNNs \cite{marciano:24}. Hence TQNNs as implementations of UQC provide a new route for engineering DL and Monte Carlo methods towards simulating scalar QFTs, and S-matrices, in light of the methods of e.g. \cite[\S V]{rinaldi:22} and \cite{hardy:24} --- cf. the groundbreaking work of \cite{jordan:14}. The correspondence between scattering and UQC may also have further bearing on e.g. the computational representations of quantum gravity effects in scattering processes such as those induced by Hawking black-hole radiation, as studied in \cite{goldberger:20}. The development of ideas leading to this correspondence, further suggests the possibility of measures of complexity corresponding to classes of amplituhedra, as having some analogy with well-known computational complexity classes (such as P, NP, and \#P, etc.).

More broadly, however, the availability of amplituhedra as a potentially tractable representation of generic quantum processes suggests applications in many areas not currently addressed by QIP.  As the edge complexity of the amplituhedron effectively measures the number of off-diagonal terms that are relevant at a given resolution or energy scale, it potentially provides a measure of how similar the QRFs employed by $A$ and $B$ are --- the extent to which they ``speak the same language'' --- in any LOCC protocol.  Consistent with the remarks above, this measure cannot be a within-protocol observable \cite{fgm:24a}.  However, general methods for addressing this question are potentially relevant to the apparent complexity difference between computation and verification (i.e. to whether, or not, does P = NP), the characterization of information flow across the connectome in neuroscience, the general question of reference frames implemented by biological networks, and other areas in which significant problems currently appear intractable.

%%\section*{Acknowledgement}

%%%%%%%%%%%%%%%%%%%%%%%%%%%%%%%%%%%%%%%%%%%%%%%%%%%%%%%%%%%%%%%%%%%%%%%%%%%%%%%%

\end{document}